%% file: two_col.tex
\documentclass[journal]{IEEEtran}
\usepackage{graphicx}
\graphicspath{{../}{./figures/}}
\usepackage[cmex10]{amsmath}
\usepackage{amssymb}
\usepackage{bm}
\usepackage{algorithm}
\usepackage{algpseudocode}
\algrenewcommand\algorithmicrequire{\textbf{Input:}}
\algrenewcommand\algorithmicensure{\textbf{Output:}}
\usepackage{array}
\usepackage{url}
\usepackage[noadjust]{cite}
\usepackage{xcolor}
\usepackage{pgfplots}
\pgfplotsset{compat=1.17}
\usetikzlibrary{decorations.markings,decorations.pathmorphing,calc}
\usetikzlibrary{positioning}
\usepgfplotslibrary{groupplots}
\newcommand{\datapath}{figures/data/}
\newif\ifonecol   
\newlength{\figw}\newlength{\figwtwo}   
\newlength{\figh}\newlength{\fightwo}   
\newlength{\figwthree}\newlength{\fighthree}   
\AtBeginDocument{\setlength{\figw}{\columnwidth}\setlength{\figwtwo}{\columnwidth}\setlength{\figh}{0.47\figw}\setlength{\fightwo}{0.41\figwtwo}\setlength{\figwthree}{0.31\textwidth}\setlength{\fighthree}{0.33\textwidth}}
\input{figures/data/values.tex}
\input{figures/data/values_mc.tex}
\def\mcNdrop{500}   
\pgfplotsset{
  cidi/.style={width=0.98\figw, height=\figh, grid=major, grid style={gray!50, dotted},
    tick label style={font=\footnotesize}, label style={font=\footnotesize}, title style={font=\footnotesize},
    legend style={font=\scriptsize, fill=white, draw=black, legend cell align=left},
    every axis plot/.append style={line width=1pt}, mark size=1.7pt, mark options={solid}, unbounded coords=jump,
    enlarge x limits=false},
  c1/.style={blue, mark=*},
  c2/.style={red, mark=square*},
  c3/.style={black, mark=triangle*},
  c4/.style={green!50!black, mark=diamond*}}
\newcommand{\fbline}[2]{\draw[#1, dotted, line width=0.8pt] ({axis cs:#2,1}|-{rel axis cs:0,0}) -- ({axis cs:#2,1}|-{rel axis cs:0,1});}
\usepackage{subcaption}
\usepackage{booktabs}
\usepackage{multirow}
\usepackage{siunitx}
\usepackage{enumerate}
\usepackage{stfloats}
\newtheorem{lemma}{Lemma}
\newtheorem{corollary}{Corollary}
\newtheorem{proposition}{Proposition}
\newtheorem{remark}{Remark}

\allowdisplaybreaks

\renewcommand{\Re}{\operatorname{Re}}
\renewcommand{\Im}{\operatorname{Im}}
\newcommand{\tr}{\operatorname{tr}}
\newcommand{\diag}{\operatorname{diag}}
\newcommand{\rk}{\operatorname{rank}}
\newcommand{\bx}{\mathbf{x}}
\newcommand{\bh}{\mathbf{h}}
\newcommand{\ba}{\mathbf{a}}
\newcommand{\bpsi}{\boldsymbol{\psi}}
\newcommand{\PEB}{\mathrm{PEB}}

\begin{document}
\title{Secure Near-Field ISAC with Pinching-Antenna Systems: Exploiting Constructive-Destructive Interference}
\author{Dimitrios~Bozanis,~\IEEEmembership{Graduate Student Member,~IEEE,} Vasilis~K.~Papanikolaou,~\IEEEmembership{Member,~IEEE,} Sotiris~A.~Tegos,~\IEEEmembership{Senior Member,~IEEE,} 
Panagiotis~D.~Diamantoulakis,~\IEEEmembership{Senior Member,~IEEE,}
Christos~Masouros,~\IEEEmembership{Fellow,~IEEE,} and George~K.~Karagiannidis,~\IEEEmembership{Fellow,~IEEE}
\thanks{D. Bozanis, S. A. Tegos, P. D. Diamantoulakis, and G. K. Karagiannidis are with the Department of Electrical and Computer Engineering, Aristotle University of Thessaloniki, 54124 Thessaloniki, Greece (e-mails: dimimpoz@ece.auth.gr, tegosoti@ece.auth.gr, padiaman@ece.auth.gr, geokarag@ece.auth.gr).}
\thanks{V. K. Papanikolaou is with the Friedrich-Alexander-Universit\"at Erlangen-N\"urnberg, 91058 Erlangen, Germany (e-mail: vasilis.papanikolaou@fau.de).}
\thanks{C. Masouros is with the Department of Electronic and Electrical Engineering, University College London, London, U.K. (e-mail: c.masouros@ucl.ac.uk).}
}
\maketitle

\begin{abstract}
This paper investigates secure near-field integrated sensing and communication (ISAC) with pinching-antenna systems (PASS), where a base station serves legitimate users while localizing a target that may also eavesdrop on the transmitted data. We jointly design symbol-level precoding (SLP) and the transmit and receive pinching-antenna positions to minimize the Cartesian position error bound (PEB) averaged over the data symbols. Constructive interference ensures reliable reception at the legitimate users, whereas destructive interference drives the target's observation toward incorrect symbol regions without sacrificing the illumination required for sensing. We derive the PEB for unknown complex target reflectivity and establish an exact factorization into transmit-illumination and receive-geometry terms. This decomposition reveals the distinct sensing roles of the two apertures, enables separate receive-placement and joint transmit-placement/SLP optimization, and provides analytical insights into local position identifiability, the range-information limitations of compact arrays, sensing-oriented receive placement, and in-waveguide attenuation. Building on these results, we develop an efficient algorithm by combining geometry-informed receive initialization, one-dimensional transmit-position searches, and convex per-symbol precoding. Numerical results demonstrate sub-millimeter localization, erroneous symbol decisions at the target, an expanded feasible operating region, and substantial gains over fixed-placement PASS and fully digital extremely large-scale MIMO (XL-MIMO) benchmarks.
\end{abstract}


\IEEEpeerreviewmaketitle

\section{Introduction}
Integrated sensing and communication (ISAC) is envisioned as a key capability of sixth-generation (6G) wireless networks, enabling data transmission and environmental sensing through shared spectrum, hardware, and signal processing resources \cite{Liu2022ISAC}. Communication must deliver information with a prescribed quality of service (QoS), while sensing requires effective target illumination and accurate parameter estimation. Sharing power and spatial degrees of freedom therefore creates a communication--sensing trade-off \cite{Liu2018MUMIMO,Liu2022CRB}.

This trade-off becomes more demanding when the sensing target can also intercept the transmitted data. The information-bearing waveform used to illuminate the target then exposes confidential symbols to a potential eavesdropper (Eve). Increasing illumination may improve sensing while strengthening the signal available for interception. Thus, security becomes a third requirement, which physical-layer security (PLS) meets by preserving reliable communication and accurate sensing while impairing the target's reception of confidential information \cite{Wei2022security}.

In conventional compact arrays, however, the approximately planar wavefront makes the spatial phase differences primarily sensitive to target direction, with limited information about distance. A sufficiently large aperture can exploit wavefront curvature to recover both parameters, enabling near-field localization \cite{NF}. To this end, pinching-antenna systems (PASS) offer a flexible means of realizing such an aperture through reconfigurable radiating and receiving points distributed along dielectric waveguides \cite{Liu2025tutorial}. Their extended spatial coverage can thus support position estimation even at distances where a compact array would operate in the far field.

\subsection{State of the Art}
Secure ISAC has been investigated through artificial noise (AN)-assisted beamforming, transmit covariance optimization, and joint communication, sensing, and jamming designs \cite{Su2021malicious,Su2024sensingassisted,Papanikolaou2026privacy}. However, allocating power to AN reduces the power available for information-bearing transmission, and AN does not guarantee an erroneous decision at the eavesdropper. Symbol-level precoding (SLP) offers an alternative by using the data symbols and channel knowledge to shape interference itself, making it beneficial for legitimate reception and harmful for eavesdropping. Constructive interference (CI) places each legitimate user's (Bob's) observation within the decision region of its intended symbol with a prescribed margin \cite{Masouros2011,Li2020tutorial}, whereas destructive interference (DI) impairs symbol detection at an unauthorized receiver (Eve) \cite{Khandaker2018}. This distinction is attractive for secure ISAC, since strong illumination of a target must not allow a favorable symbol decision for eavesdropping. Interference exploitation has been applied to ISAC in \cite{Wang2025uniform} and to secure ISAC with eavesdropping targets in \cite{Su2022secureDFRC,Jia2026,Su2025BCRB,Bozanis2026robust}. These studies rely on fixed-position arrays and focus on target direction or spatial illumination, without investigating how antenna reconfiguration can enable joint direction and distance estimation and potentially satisfy communication and security requirements that are infeasible with a fixed array.

PASS offer a promising way to explore these possibilities by allowing the antenna positions, and hence the propagation, to be adjusted, opening opportunities to manage the trade-offs among communication, sensing, and security. First demonstrated by NTT DOCOMO \cite{Suzuki2022}, the concept uses small dielectric elements, termed pinching antennas (PAs), to couple energy between a waveguide and free space. Early analysis established their potential to reduce path loss and manage multi-user interference \cite{Ding2024flexible}, while \cite{Liu2025tutorial} provided a systematic treatment of PASS modeling and beamforming. Unlike movable and fluid antennas, PASS support position adjustments over several meters, affecting both the phase and the path loss of the channel.

For communication, \cite{Xu2025rate} developed a two-stage placement method that balances path-loss reduction and phase alignment to maximize the downlink rate. The analysis in \cite{Tyrovolas2025} characterized outage probability, average rate, and optimal placement under waveguide attenuation. Building on a physical model of antenna coupling, \cite{Wang2025modeling} jointly optimized transmit beamforming and PA activation to minimize transmit power.

Regarding PASS sensing, \cite{Bozanis2025CRB} derived closed-form Cram\'er--Rao bounds (CRBs) for bistatic joint range and direction estimation, while \cite{Li2025PASSISAC} jointly designed beamforming and PA placement to minimize the sensing CRB under communication constraints. Target uncertainty was considered in \cite{Jiang2025BCRB} through Bayesian bounds and placement strategies based on prior distributions. The Ziv--Zakai analysis in \cite{Jiang2026ZZB} further accounted for localization ambiguities that local bounds may not capture. Moreover, \cite{Chen2026unified} proposed a common framework supporting uplink and downlink communication with active and passive targets. Finally, \cite{Feng2026phase} derived the Cartesian position error bound (PEB) for uplink localization by exploiting both amplitude and phase information.

For secure PASS transmission, \cite{Sun2025PLS} jointly optimized baseband beamforming and PA positions to improve secrecy rates. AN-assisted transmission under imperfect channel state information (CSI) of Eve was investigated in \cite{Papanikolaou2025AN}, while \cite{Zhang2026attenuation} derived secrecy outage and ergodic secrecy capacity bounds accounting for waveguide losses. Bringing security and sensing together, \cite{Illi2025} optimized PA placement, information beamforming, and AN to improve target illumination under secrecy constraints. More recently, \cite{Song2026} combined transmit PAs with leaky coaxial receivers to track Eve's position and velocity, updating PA configurations and baseband processing on different timescales. As discussed above, stronger illumination then also increases the leakage, whereas SLP can support illumination while impairing eavesdropping. For PASS, \cite{Pang2026} jointly optimized SLP and PA positions to minimize transmit power, but without considering sensing or security requirements.

\subsection{Motivation and Contributions}
The above studies leave unexplored how SLP and PA placement can be jointly designed for secure Cartesian localization, and whether their combination can expand the feasible operating region under communication and security constraints. Equally important is understanding which antenna positions best support localization and how the transmit and receive apertures contribute to sensing performance. Positions that strengthen the communication links or increase target illumination do not necessarily provide the spatial information needed for accurate position estimation. This raises the question of how the two apertures should be configured and how in-waveguide attenuation affects this choice.

Motivated by these questions, this paper investigates joint SLP and antenna placement in a downlink PASS-ISAC system, where a base station (BS) serves multiple Bobs and localizes a point target that is also a potential passive Eve. We jointly design the transmit waveforms and the PA positions on both sides, while establishing how illumination and receive geometry influence sensing accuracy. The resulting framework combines CI for the Bobs and DI toward the target/Eve with a sensing objective that directly quantifies Cartesian position accuracy. The main contributions are summarized as follows:
\begin{itemize}
\item We derive the Cartesian PEB of a PASS-ISAC system for target localization with unknown complex reflectivity and establish an exact factorization that separates receive geometry from transmit illumination. Under the dominant phase-information approximation, we further characterize the receive-geometry requirements for local position estimation and identify antenna configurations that make localization impossible.
\item We analytically characterize the range-information limitations of compact linear receive arrays and obtain a closed-form characterization of sensing-optimal PA offsets for a symmetric receive layout. These results explain why proximity to the target alone is insufficient and guide the initialization of more general deployments. We also quantify the impact of in-waveguide attenuation on localization accuracy.
\item We formulate the joint design of SLP and transmit/receive PA positions to minimize the Cartesian PEB averaged over data-symbol vectors. The formulation incorporates CI for the Bobs, DI toward the target/Eve, and individual transmit-waveguide power limits, with PA positions shared across symbol intervals.
\item We exploit the PEB factorization to separate receive PA placement from the joint transmit placement and SLP design. The receive positions are optimized through a cyclic search initialized using the derived geometric insights, while the transmit-side design combines one-dimensional position searches with convex per-symbol precoding subproblems over the wrong-symbol regions, each of which is solved in closed form. We also establish convergence and derive the complexity of the proposed scheme.
\item Numerical results validate the analysis and demonstrate improved localization performance and an expanded feasible operating region under CI/DI requirements compared with fixed-placement PASS and a compact array. They also illustrate the distinct roles of the transmit and receive apertures, the impact of in-waveguide attenuation, and the trade-off between localization accuracy and communication QoS.
\end{itemize}

\subsection{Structure and Notation}
The remainder of this paper is organized as follows. Section~\ref{sec:model} presents the system model. Section~\ref{sec:metric} derives the Cartesian PEB and establishes the geometric properties relevant to PA placement. Section~\ref{sec:design} formulates the joint design problem and develops the proposed solution. Section~\ref{sec:results} presents the numerical results, and Section~\ref{sec:conclusion} concludes the paper.

Notation: Bold lowercase and uppercase letters denote vectors and matrices, respectively, and calligraphic letters denote sets. The operators $(\cdot)^{\top}$, $(\cdot)^H$, and $(\cdot)^*$ denote transpose, Hermitian transpose, and complex conjugation, respectively. Element-wise division is denoted by $\oslash$, and $|\mathbf{a}|^{\circ 2}$ denotes the vector containing the squared magnitudes of the entries of $\mathbf{a}$. The operators $\tr(\cdot)$ and $\rk(\cdot)$ denote trace and rank, respectively. The notation $[\mathbf{A}]_{i,j}$ denotes the $(i,j)$-th entry of $\mathbf{A}$, and $[\mathbf{A}]_{1:2,1:2}$ denotes its leading $2\times2$ block. The real and imaginary parts are denoted by $\Re\{\cdot\}$ and $\Im\{\cdot\}$, respectively. The diagonal matrix $\diag(\mathbf{w})$ has the entries of $\mathbf{w}$ on its diagonal. The all-ones vector and the $N\times N$ identity matrix are denoted by $\mathbf{1}$ and $\mathbf{I}_N$, respectively. The notation $\|\cdot\|$ denotes the Euclidean norm, and $\angle s$ denotes the phase of $s$. Expectation is denoted by $\mathbb{E}[\cdot]$, with a subscript indicating the random variable. Finally, $\mathcal{CN}(\mathbf{0},\mathbf{R})$ denotes the circularly symmetric complex Gaussian distribution with zero mean and covariance matrix $\mathbf{R}$.

\section{System Model}\label{sec:model}
\begin{figure*}[t]\centering
\includegraphics[width=0.9\textwidth, trim={0 45 55 5}, clip]{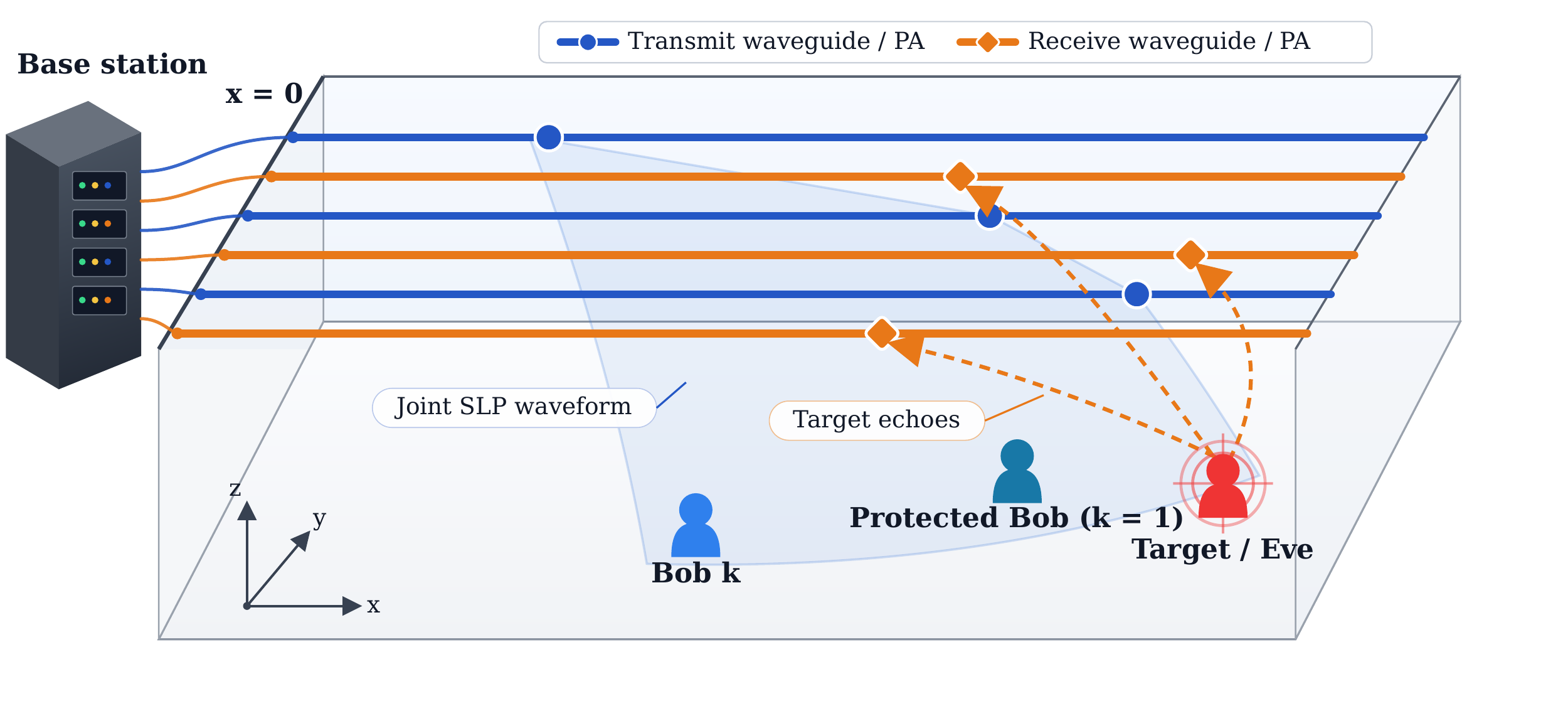}
\caption{The considered secure PASS-ISAC system.}\label{fig:system}
\end{figure*}
\subsection{Pinching-Antenna Model and Network Topology}
We consider the indoor downlink ISAC system illustrated in Fig.~\ref{fig:system}. The BS is connected to $N_{\rm T}$ transmit and $N_{\rm R}$ receive dielectric waveguides, indexed by the sets $\mathcal{N}_{\rm T}=\{1,\dots,N_{\rm T}\}$ and $\mathcal{N}_{\rm R}=\{1,\dots,N_{\rm R}\}$, respectively, all of length $D_x$, parallel to the $x$-axis, mounted at height $d$, and fed at $x=0$. The $n$-th transmit waveguide, $n\in\mathcal{N}_{\rm T}$, lies at $y=\tilde y^{\rm T}_n$ and carries a single PA at $\tilde\bpsi^{\rm T}_n=(\tilde x^{\rm T}_n,\tilde y^{\rm T}_n,d)$, and the $m$-th receive waveguide, $m\in\mathcal{N}_{\rm R}$, lies at $y=\tilde y^{\rm R}_m$ and carries a single PA at $\tilde\bpsi^{\rm R}_m=(\tilde x^{\rm R}_m,\tilde y^{\rm R}_m,d)$. Each transmit waveguide is driven by one RF chain and each receive waveguide feeds one receive chain. The transmit and the receive waveguides are separate, so that the receive chains are free of self-interference, as is standard for co-located ISAC transceivers with separate transmit and receive apertures \cite{Liu2018MUMIMO}. In a PASS, the two sets of waveguides are distinct dielectric structures, fed independently and spaced apart in the ceiling, so that the residual coupling is a fixed leakage that can be calibrated.
The system serves $K$ single-antenna legitimate users (Bobs), indexed by the set $\mathcal{K}=\{1,\dots,K\}$ and located at $\bpsi_k=(x_k,y_k,0)$, $k\in\mathcal{K}$, and a point target, which is also a potential passive eavesdropper (Eve), is located at $\bpsi_t=(x_t,y_t,0)$.

The channel from a PA at $\tilde\bpsi=(\tilde x,\tilde y,d)$ to a point $\bpsi$ on the ground comprises the free-space spherical-wave propagation, the in-waveguide phase shift accumulated from the feed point to the PA, and the in-waveguide attenuation \cite{Xu2025rate,Tyrovolas2025}, and is given by
\begin{equation}\label{eq:pa_channel}
g(\bpsi;\tilde\bpsi)=\frac{\sqrt{\eta}\,e^{-\alpha\tilde x/2}\,e^{-j\left(\kappa\|\bpsi-\tilde\bpsi\|+\kappa_g\tilde x\right)}}{\|\bpsi-\tilde\bpsi\|},
\end{equation}
where $\eta=(\lambda/4\pi)^2$ denotes the free-space channel gain at a reference distance of $1$ m, $\lambda$ is the free-space wavelength, $\kappa=2\pi/\lambda$ is the free-space wavenumber, and $\kappa_g=n_{\rm eff}\kappa$ is the guided wavenumber, with $n_{\rm eff}$ denoting the effective refractive index of the dielectric waveguide. Moreover, $\alpha$ [Np/m] is the power attenuation coefficient of the waveguide, so that the guided power decays as $e^{-\alpha\tilde x}$ along the waveguide \cite{Tyrovolas2025}. The transmit and receive channel vectors to a point $\bpsi=(x,y,0)$ are collected as
\begin{align}
\mathbf{g}(\bpsi)&=\left[g(\bpsi;\tilde\bpsi^{\rm T}_1),\dots,g(\bpsi;\tilde\bpsi^{\rm T}_{N_{\rm T}})\right]^{\top}\in\mathbb{C}^{N_{\rm T}},\label{eq:gvec}\\
\ba(\bpsi)&=\left[g(\bpsi;\tilde\bpsi^{\rm R}_1),\dots,g(\bpsi;\tilde\bpsi^{\rm R}_{N_{\rm R}})\right]^{\top}\in\mathbb{C}^{N_{\rm R}},\label{eq:avec}
\end{align}
and we define the $k$-th Bob, target/Eve, and receive channels as
\begin{equation}\label{eq:channels}
\bh_k=\mathbf{g}(\bpsi_k)^*,\qquad \bh_e=\mathbf{g}(\bpsi_t)^*,\qquad \ba_t=\ba(\bpsi_t).
\end{equation}
The transmit and receive PA positions, collected in
\begin{equation}\label{eq:positions}
\tilde{\bx}^{\rm T}=\left[\tilde x^{\rm T}_1,\dots,\tilde x^{\rm T}_{N_{\rm T}}\right]^{\top},\qquad
\tilde{\bx}^{\rm R}=\left[\tilde x^{\rm R}_1,\dots,\tilde x^{\rm R}_{N_{\rm R}}\right]^{\top},
\end{equation}
are the reconfigurable variables of the PASS and affect the model only through the channel vectors $\mathbf{g}(\cdot)$ and $\ba(\cdot)$, respectively.

\subsection{Communication Signal Model}
Let $\mathbf{s}=[s_1,\dots,s_K]^{\top}$ denote the $M$-PSK symbols intended for the $K$ Bobs, with $M$ the constellation order. The BS generates a symbol-dependent transmit vector $\bx\in\mathbb{C}^{N_{\rm T}}$, whose $n$-th entry feeds the $n$-th transmit waveguide. Since each waveguide has its own power amplifier, the power is limited per waveguide,
\begin{equation}\label{eq:power}
|x_n|^2\le P/N_{\rm T},\qquad \forall n\in\mathcal{N}_{\rm T},
\end{equation}
where $P$ is the total power budget. The received signal at Bob $k$ is given by
\begin{equation}\label{eq:yk}
y_k=\bh_k^H\bx+n_k,\qquad \forall k\in\mathcal{K},
\end{equation}
where $n_k\sim\mathcal{CN}(0,\sigma_k^2)$ is the additive white Gaussian noise (AWGN) at Bob $k$. It is convenient to express the noise-free part of \eqref{eq:yk} in a coordinate system rotated by the phase of $s_k$ \cite{Li2020tutorial}, i.e.,
\begin{equation}\label{eq:ytk}
\tilde y_k=\bh_k^H\bx\,e^{-j\angle s_k}.
\end{equation}

\subsection{Eavesdropping Model}
The signal observed by the target/Eve is
\begin{equation}\label{eq:ye}
y_e=\bh_e^H\bx+n_e,
\end{equation}
where $n_e\sim\mathcal{CN}(0,\sigma_e^2)$ is the AWGN at the target/Eve. The target/Eve attempts to decode the confidential stream of the protected Bob, which we take to be the Bob nearest to the target/Eve \cite{Su2022secureDFRC,Jia2026,Bozanis2026robust}. Without loss of generality, the protected Bob is indexed by $k=1$, and, as in \eqref{eq:ytk}, the rotated noise-free observation of the target/Eve is $\tilde y_e=\bh_e^H\bx\,e^{-j\angle s_1}$.

\subsection{Sensing Signal Model}
The echo of the target is collected by the receive PAs, guided to the feeds of the $N_{\rm R}$ receive waveguides, and delivered to the receive chains. Following the uplink PASS model of \cite{Tegos2025uplink}, the echo received on the $m$-th receive chain is proportional to $g(\bpsi_t;\tilde\bpsi^{\rm R}_m)$, so that the received signal vector is given by
\begin{equation}\label{eq:echo}
\mathbf{y}_{\rm R}=\alpha_t\,\ba_t\,\bh_e^H\bx+\mathbf{n}_{\rm R},\qquad \mathbf{n}_{\rm R}\sim\mathcal{CN}(\mathbf{0},\sigma_{\rm R}^2\mathbf{I}_{N_{\rm R}}),
\end{equation}
where $\mathbf{n}_{\rm R}$ is the AWGN vector of the receive chains and $\alpha_t\in\mathbb{C}$ is the unknown complex reflectivity of the target. We define the illumination of the target as
\begin{equation}\label{eq:illum}
p\triangleq\bh_e^H\bx,
\end{equation}
i.e., the noise-free part of the observation \eqref{eq:ye} of the target/Eve, whose magnitude $|p|=|\tilde y_e|$ is unaffected by the rotation by $\angle s_1$. The target/Eve is modeled as a static point scatterer with radar cross section $\sigma_{\rm RCS}$. Since the free-space losses of the transmit and receive paths are already contained in $\bh_e$ and $\ba_t$, the radar equation gives $|\alpha_t|=\sqrt{4\pi\sigma_{\rm RCS}}/\lambda$ \cite{Liu2022CRB,Bozanis2025CRB,Su2025BCRB}, which is used to evaluate the bound, while the real and imaginary parts of $\alpha_t$ are treated as unknown nuisance parameters in the estimation.
\subsection{CI and DI Regions}
\begin{figure*}[t]\centering
{\setlength{\figwtwo}{0.85\figwtwo}\input{figures/figCIDI.tex}\hspace{0.03\textwidth}\input{figures/figC.tex}}
\caption{(a) CI and DI regions for QPSK and (b) received symbols with a zoom around the origin.}\label{fig:cidi}
\end{figure*}
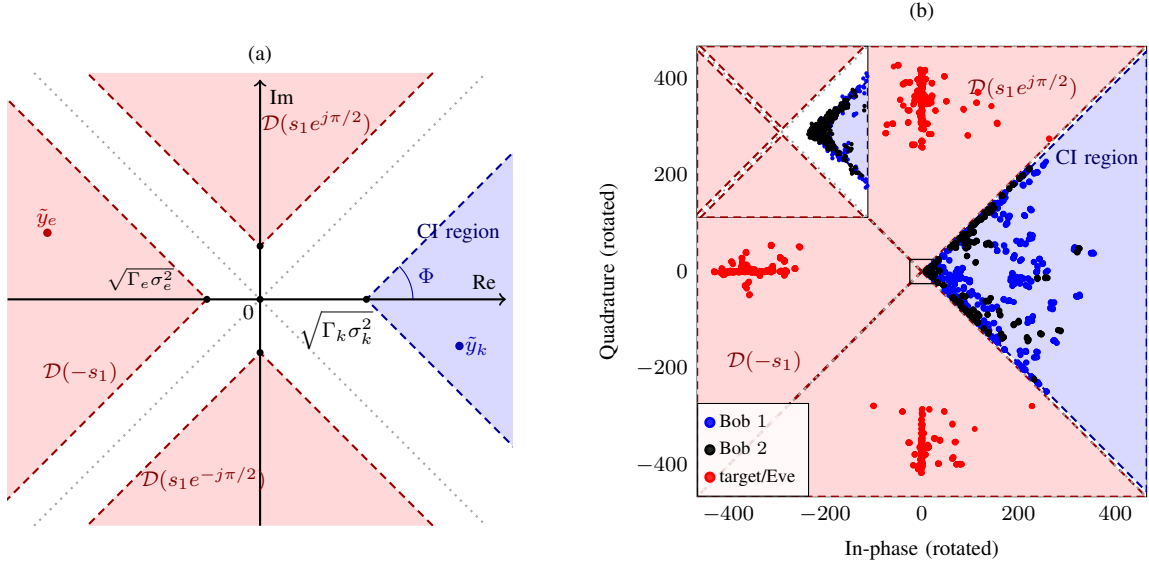
For $M$-PSK signaling with $M>2$, the decision region of each symbol is a sector of half-angle $\Phi=\pi/M$ centered at the symbol. The CI requirement of Bob $k$ states that $\tilde y_k$ lies inside the decision region of $s_k$ with a safety margin determined by the SNR threshold $\Gamma_k$ \cite{Li2020tutorial}, as illustrated in Fig.~\ref{fig:cidi}(a) for QPSK, i.e.,
\begin{equation}\label{eq:CI}
\left|\Im\{\tilde y_k\}\right|\le\left(\Re\{\tilde y_k\}-\sqrt{\Gamma_k\sigma_k^2}\right)\tan\Phi,\qquad \forall k\in\mathcal{K}.
\end{equation}
By decomposing \eqref{eq:CI} along the two boundary rays of the decision sector, condition \eqref{eq:CI} is equivalent to the pair of linear inequalities
\begin{equation}\label{eq:CIlin}
\Re\left\{\tilde\bh_{k,i}^H\bx\right\}\ge\gamma_k,\qquad i\in\{1,2\},\ \forall k\in\mathcal{K},
\end{equation}
where the rotated channels and the threshold are given by
\ifonecol
\begin{align}
\tilde\bh_{k,1}&=e^{j\angle s_k}\left(\sin\Phi+j\cos\Phi\right)\bh_k,\quad \tilde\bh_{k,2}=e^{j\angle s_k}\left(\sin\Phi-j\cos\Phi\right)\bh_k,\label{eq:hk12}\\
\gamma_k&=\sqrt{\Gamma_k\sigma_k^2}\,\sin\Phi.\label{eq:gammak}
\end{align}
\else
\begin{align}
\tilde\bh_{k,1}&=e^{j\angle s_k}\left(\sin\Phi+j\cos\Phi\right)\bh_k,\nonumber\\
\tilde\bh_{k,2}&=e^{j\angle s_k}\left(\sin\Phi-j\cos\Phi\right)\bh_k,\label{eq:hk12}\\
\gamma_k&=\sqrt{\Gamma_k\sigma_k^2}\,\sin\Phi.\label{eq:gammak}
\end{align}
\fi

For the target/Eve, the DI requirement states that $\tilde y_e$ lies outside the decision region of the protected symbol $s_1$ and inside the destructive region $\mathcal{D}_{\rm DI}$, which is formed by the decision regions of the $M-1$ wrong symbols, each with a detection margin $\Gamma_e$, i.e., the union of the red cones in Fig.~\ref{fig:cidi}(a). The margin plays for the target/Eve the same role that $\Gamma_k$ plays for the Bobs, but with respect to a wrong symbol, so that the target/Eve detects $s'$ with an SNR of at least $\Gamma_e$, i.e., it decides a wrong symbol with high probability. Let $\mathcal{D}(s')$ denote the set of transmit vectors for which the observation $\bh_e^H\bx$ lies inside the decision region of a wrong symbol $s'\ne s_1$ with margin $\sqrt{\Gamma_e\sigma_e^2}$, i.e., a convex cone shifted away from the origin by the margin
\begin{equation}\label{eq:DI}
\mathcal{D}(s')=\left\{\bx:\ \Re\{\tilde\ba_{i}^H(s')\bx\}\ge\gamma_e,\ i=1,2\right\},
\end{equation}
where
\ifonecol
\begin{align}
\tilde\ba_{1}(s')&=e^{j\angle s'}\left(\sin\Phi+j\cos\Phi\right)\bh_e,\quad \tilde\ba_{2}(s')=e^{j\angle s'}\left(\sin\Phi-j\cos\Phi\right)\bh_e,\label{eq:ae12}\\
\gamma_e&=\sqrt{\Gamma_e\sigma_e^2}\,\sin\Phi.\label{eq:gammae}
\end{align}
\else
\begin{align}
\tilde\ba_{1}(s')&=e^{j\angle s'}\left(\sin\Phi+j\cos\Phi\right)\bh_e,\nonumber\\
\tilde\ba_{2}(s')&=e^{j\angle s'}\left(\sin\Phi-j\cos\Phi\right)\bh_e,\label{eq:ae12}\\
\gamma_e&=\sqrt{\Gamma_e\sigma_e^2}\,\sin\Phi.\label{eq:gammae}
\end{align}
\fi
Then, the DI region is the union
\begin{equation}\label{eq:DIunion}
\mathcal{D}_{\rm DI}=\bigcup_{s'\ne s_1}\mathcal{D}(s').
\end{equation}
For QPSK, the wrong symbols are $s_1e^{\pm j\pi/2}$ and $-s_1$, and each apex in Fig.~\ref{fig:cidi}(a) lies at distance $\sqrt{\Gamma_e\sigma_e^2}$ from the origin along its wrong symbol.

\section{Sensing Metric and Geometric Insights}\label{sec:metric}
\subsection{Cartesian Position Error Bound}
The sensing parameter of interest is the position $\bpsi_t$ of the target on the ground, and the sensing performance is quantified through the PEB, i.e., the CRB, which lower-bounds the root-mean-square error of any unbiased estimate of $\bpsi_t$. The following lemma gives the single-snapshot PEB of the considered secure PASS-ISAC system in closed form.

\begin{lemma}\label{th:factor}
Let $\boldsymbol{\xi}=[x_t,y_t,\Re\{\alpha_t\},\Im\{\alpha_t\}]^{\top}$ collect the unknowns of \eqref{eq:echo}, with the complex reflectivity treated as a nuisance parameter. For a single symbol interval, the equivalent Fisher information matrix (FIM) of the target position is
\begin{equation}\label{eq:Je}
\mathbf{J}_{\rm e}=\frac{2|\alpha_t|^2}{\sigma_{\rm R}^2}\left|\bh_e^H\bx\right|^2\mathbf{G}(\tilde{\bx}^{\rm R}),
\end{equation}
where $\mathbf{G}(\tilde{\bx}^{\rm R})$ is the geometric term, which depends only on the receive PA positions and is given by
\begin{equation}\label{eq:G}
\mathbf{G}(\tilde{\bx}^{\rm R})=\Re\left\{\mathbf{D}^H\mathbf{W}\mathbf{D}\right\},
\end{equation}
with
\begin{equation}\label{eq:W}
\mathbf{W}=\diag(\mathbf{w})-\frac{\mathbf{w}\mathbf{w}^{\top}}{\mathbf{1}^{\top}\mathbf{w}},\qquad \mathbf{w}=|\ba_t|^{\circ 2},
\end{equation}
being the centered matrix of the receive power weights and
\begin{equation}\label{eq:D}
\mathbf{D}=\left[\frac{\partial\ba_t}{\partial x_t}\oslash\ba_t,\ \frac{\partial\ba_t}{\partial y_t}\oslash\ba_t\right]\in\mathbb{C}^{N_{\rm R}\times2}
\end{equation}
the matrix of logarithmic derivatives of the receive channel with respect to the target position, given in closed form in Appendix~\ref{app:gdop}. The PEB of the target position can therefore be expressed as
\begin{equation}\label{eq:PEBfactor}
\PEB=\sqrt{\tr\left(\mathbf{J}_{\rm e}^{-1}\right)}=\frac{\overbrace{c(\tilde{\bx}^{\rm R})}^{\text{geometry}}}{\underbrace{\left|\bh_e^H\bx\right|}_{\text{illumination}}},\quad c(\tilde{\bx}^{\rm R})=\sqrt{\frac{\sigma_{\rm R}^2\,\tr\left(\mathbf{G}^{-1}\right)}{2|\alpha_t|^2}}.
\end{equation}
\end{lemma}
\begin{IEEEproof}
The proof is presented in Appendix~\ref{app:factor}.
\end{IEEEproof}

Lemma~\ref{th:factor} leads to an important observation on how the three requirements of the system act on the transmit vector. For given receive PA positions, the sensing performance depends on $\bx$ only through the magnitude $|p|$ of the target observation $p=\bh_e^H\bx$, and not on its phase. The DI constraint \eqref{eq:DIunion} confines the phase of $p$ to a wrong-symbol sector and imposes a minimum decision margin, since it requires $\tilde y_e=p\,e^{-j\angle s_1}$ to lie inside one of the regions $\mathcal{D}(s')$ of Fig.~\ref{fig:cidi}(a), but does not upper-bound $|p|$, whereas the CI constraints \eqref{eq:CI} involve the observations $\bh_k^H\bx$ of the Bobs and not $p$. Moreover, scaling a feasible observation outward preserves feasibility and drives it deeper into the wrong-symbol region, so that the probability of an erroneous decision at the target/Eve approaches one as the illumination increases. This is not the case under the DI condition of \cite{Su2022secureDFRC,Jia2026,Su2025BCRB,Bozanis2026robust}, in which $\Gamma_e$ upper-bounds the SNR of the correct symbol, so that observations inside the decision region of $s_1$ but below the margin are admissible, and scaling such an observation moves it into the CI region of $s_1$. Hence, unlike AN-based designs, in which the target illumination is the secrecy leak, CI/DI precoding with \eqref{eq:DIunion} allows the target to be illuminated at full power without granting it access to the confidential symbols.

\begin{remark}\label{rem:txaperture}
By Lemma~\ref{th:factor}, the structure of $\mathbf{J}_{\rm e}$ is determined only by the receive PA positions, since any variation of the illumination $p$ with $\bpsi_t$, in phase or magnitude, is indistinguishable from a variation of the unknown $\alpha_t$ in \eqref{eq:echo} and carries no position information, as shown in \eqref{eq:PiV}. Hence, no transmit array can compensate for an uninformative receive geometry. However, the transmit side remains essential, as it sets the scale of the PEB through $|p|$ and is the only side that affects the CI and DI constraints, i.e., it governs the sensing SNR and feasibility.
\end{remark}

\subsection{Geometric Insights}\label{sec:geom}
Lemma~\ref{th:factor} confines the effect of the receive PA positions to the geometric term $\mathbf{G}$, but \eqref{eq:G} does not reveal how these positions should be chosen. Having established the factorization of $\PEB$, we now turn our attention to the structure of $\mathbf{G}$ and express it in terms of the geometry between the target and the receive PAs, from which the conditions for the identifiability of the target position and the placement rules for the receive PAs arise.

\begin{lemma}\label{lem:gdop}
Let $\mathbf{u}_m=\left[\tilde x^{\rm R}_m-x_t,\ \tilde y^{\rm R}_m-y_t\right]^{\top}/r_m$, with $r_m=\|\bpsi_t-\tilde\bpsi^{\rm R}_m\|$, collect the horizontal components of the unit vector from the target to the $m$-th receive PA, referred to as the bearing vector. Then, up to terms of relative order $(\kappa r_m)^{-2}$,
\begin{equation}\label{eq:gdop}
\mathbf{G}(\tilde{\bx}^{\rm R})=\kappa^2\sum_{m=1}^{N_{\rm R}}w_m\left(\mathbf{u}_m-\bar{\mathbf{u}}\right)\left(\mathbf{u}_m-\bar{\mathbf{u}}\right)^{\top},
\end{equation}
where $w_m$ is the $m$-th entry of $\mathbf{w}$ in \eqref{eq:W} and
\begin{equation}\label{eq:ubar}
\bar{\mathbf{u}}=\frac{\sum_{m=1}^{N_{\rm R}} w_m\mathbf{u}_m}{\sum_{m=1}^{N_{\rm R}} w_m},
\end{equation}
is the power-weighted mean of the bearing vectors, i.e., the geometric term is $\kappa^2$ times the power-weighted covariance of the bearing vectors.
\end{lemma}
\begin{IEEEproof}
The proof is presented in Appendix~\ref{app:gdop}.
\end{IEEEproof}

Lemma~\ref{lem:gdop} has a significant interpretation. A displacement $\delta\bpsi$ of the target changes the distance to the $m$-th receive PA by $-\mathbf{u}_m^{\top}\delta\bpsi$, and hence the phase of its echo by $\kappa\,\mathbf{u}_m^{\top}\delta\bpsi$, which is the quantity the receiver measures. However, a phase shift that is common to all receive PAs is indistinguishable from the unknown phase of $\alpha_t$, so that only the differences between the phase shifts of the receive PAs carry position information. This is the role of the centering by $\bar{\mathbf{u}}$ in \eqref{eq:gdop}. Consequently, $\mathbf{G}$ grows with the spread of the bearing vectors, each weighted by the power received by its PA, and it is singular when all of them are parallel. In this sense, Lemma~\ref{lem:gdop} is the near-field counterpart of the geometric dilution of precision (GDOP) of multilateration, whose anchors are fixed, whereas the weights $w_m$ and the bearings $\mathbf{u}_m$ of a PASS are design variables.


\begin{proposition}\label{prop:ident}
Let $\mathbf{u}_1,\dots,\mathbf{u}_{N_{\rm R}}$ be the bearing vectors of Lemma~\ref{lem:gdop}, and let $\mathbf{G}$ be given by \eqref{eq:gdop}, i.e., up to terms of relative order $(\kappa r_m)^{-2}$. Then, $\rk(\mathbf{G})\le N_{\rm R}-1$, and $\mathbf{G}$ is nonsingular if and only if the bearing vectors do not lie on a common line. Hence, estimating the two-dimensional target position requires $N_{\rm R}\ge3$ receive waveguides with non-collinear bearing vectors. In particular, if all receive PAs are placed at $\tilde x^{\rm R}_m=x_t$, then $x_t$ is unobservable.
\end{proposition}
\begin{IEEEproof}
By \eqref{eq:gdop}, $\mathbf{G}$ is a weighted sum of the rank-one matrices $(\mathbf{u}_m-\bar{\mathbf{u}})(\mathbf{u}_m-\bar{\mathbf{u}})^{\top}$, whose centered vectors satisfy $\sum_{m}w_m(\mathbf{u}_m-\bar{\mathbf{u}})=\mathbf{0}$ by the definition \eqref{eq:ubar} of $\bar{\mathbf{u}}$. At most $N_{\rm R}-1$ of these vectors are therefore linearly independent, so that $\rk(\mathbf{G})\le N_{\rm R}-1$, and $\mathbf{G}$ has full rank two if and only if the centered vectors span the plane, i.e., if and only if the points $\mathbf{u}_m$ do not lie on a common line. Since two points always lie on a line, $N_{\rm R}\ge3$ is necessary. Finally, if $\tilde x^{\rm R}_m=x_t$ for all $m$, then $\mathbf{u}_m=[0,\ (\tilde y^{\rm R}_m-y_t)/r_m]^{\top}$ for all $m$, so that the first entry of every centered vector $\mathbf{u}_m-\bar{\mathbf{u}}$ is zero and, by \eqref{eq:gdop}, $[\mathbf{G}]_{x_tx_t}=0$, i.e., $x_t$ is unobservable.
\end{IEEEproof}

With $N_{\rm R}=2$, the amplitude term of \eqref{eq:ReDWD}, neglected in \eqref{eq:gdop}, keeps the position identifiable when the two PAs are at different distances from the target, but with a PEB larger by a factor of order $\kappa r_m$, so that the position is practically unobservable, as for the compact arrays considered next.

\begin{proposition}\label{prop:compact}
Consider a conventional receive array of $N$ elements placed along a line at positions $q_n$, $n=1,\dots,N$, with aperture $A$, and let the target lie at broadside, at a distance $R\gg A$. Let the range direction be the line of sight to the target and the cross-range direction be parallel to the array, and let $[\mathbf{G}]_{\rm range}$ and $[\mathbf{G}]_{\rm cross}$ be the entries of $\mathbf{G}$ associated with the range and cross-range displacements of the target. Then, up to terms of relative order $(\kappa R)^{-2}$,
\begin{equation}\label{eq:compact}
\frac{[\mathbf{G}]_{\rm range}}{[\mathbf{G}]_{\rm cross}}=\frac{\mathrm{Var}\{q_n^2\}}{4R^2\,\mathrm{Var}\{q_n\}}\;\xrightarrow{N\to\infty}\;\frac{A^2}{60R^2},
\end{equation}
where $\mathrm{Var}\{\cdot\}$ denotes the empirical variance over the $N$ elements and the limit corresponds to a uniform aperture. Hence, the information on the range of the target is smaller than the information on its cross-range position by a factor of order $A^2/R^2$, so that the PEBs of the range and the cross-range coordinates of a uniform aperture satisfy
\begin{equation}\label{eq:PEBratio}
\PEB_{\rm range}\approx\sqrt{60}\,\frac{R}{A}\,\PEB_{\rm cross},
\end{equation}
independently of the SNR and of the number of snapshots.
\end{proposition}
\begin{IEEEproof}
For a broadside target at $(0,R)$ and elements at $(q_n,0)$, a second-order expansion of the bearing vectors of Lemma~\ref{lem:gdop} gives $u_{x,n}\approx q_n/R$ and $u_{y,n}\approx-1+q_n^2/(2R^2)$, so that, with equal weights $w$, \eqref{eq:gdop} yields $[\mathbf{G}]_{\rm cross}=N\kappa^2w\,\mathrm{Var}\{q_n\}/R^2$, $[\mathbf{G}]_{\rm range}=N\kappa^2w\,\mathrm{Var}\{q_n^2\}/(4R^4)$, and a zero cross term by symmetry. Their ratio is \eqref{eq:compact}, and the limit follows from $\mathrm{Var}\{q\}=A^2/12$ and $\mathrm{Var}\{q^2\}=A^4/180$ for a uniform aperture.
\end{IEEEproof}

Proposition~\ref{prop:compact} also justifies the spherical-wave model \eqref{eq:pa_channel}, since a $10$ m aperture at $28$ GHz has a Fraunhofer distance of about $19$ km. Conversely, it explains why secure ISAC designs with conventional arrays estimate the target angle only. For a $\lambda/2$-spaced array of $4$ elements at $28$ GHz and $R=4$ m, the range PEB is about $1700$ times the cross-range PEB. This limitation cannot be lifted by more snapshots or by the transmit side, as indicated in Remark~\ref{rem:txaperture}, and applies equally to movable and fluid antennas. A PASS, in contrast, obtains bearing vectors that differ by $\mathcal{O}(1)$ across as few as three receive waveguides.

The in-waveguide attenuation affects the PEB as follows. Since the factor $e^{-\alpha x_t/2}$ is common to all channels in \eqref{eq:pa_channel}, the illumination scales as $e^{-\alpha x_t/2}$, the weights in \eqref{eq:W} scale as $e^{-\alpha x_t}$, and $\mathbf{D}$ in \eqref{eq:D} depends only on the positions of the PAs relative to the target. Hence, the optimal relative positions do not depend on $x_t$, and the PEB satisfies
\begin{equation}\label{eq:pebexp}
\PEB\propto e^{\alpha x_t},
\end{equation}
i.e., the attenuation is paid twice, on the way to the transmit PAs and on the way back from the receive PAs, and the PEB doubles every $\ln 2/\alpha$ meters.

\section{Problem Formulation and Proposed Solution}\label{sec:design}
In the secure SLP designs of \cite{Su2022secureDFRC,Jia2026,Su2025BCRB,Bozanis2026robust}, the transmit vector is the only design variable. Here, the PA positions on both the transmit and the receive side are also design variables. Since activating a PA at a new point is a mechanical operation, the practical mode of operation of a PASS is to hold the PA positions fixed over many symbol intervals, whereas the transmit vector and, with it, the PEB change with each symbol vector. Therefore, placement must account for every symbol vector that may be transmitted, and the design objective is the PEB averaged over the $M^K$ equiprobable symbol vectors. Hence, the design problem is formulated as
\begin{equation}
\begin{aligned}
\min_{\tilde{\bx}^{\rm T},\,\tilde{\bx}^{\rm R},\,\bx(\mathbf{s})}\quad&\mathbb{E}_{\mathbf{s}}\left[\PEB\right]\\
\text{s.t.}\quad
&\mathrm{C}_1:\ \Re\{\tilde\bh_{k,i}^H\bx\}\ge\gamma_k,\ i\in\{1,2\},\ \forall k\in\mathcal{K},\\
&\mathrm{C}_2:\ \bx\in\mathcal{D}_{\rm DI},\\
&\mathrm{C}_3:\ |x_n|^2\le P/N_{\rm T},\ \forall n\in\mathcal{N}_{\rm T},\\
&\mathrm{C}_4:\ \tilde x^{\rm T}_n,\ \tilde x^{\rm R}_m\in[0,D_x],\ \forall n\in\mathcal{N}_{\rm T},\ m\in\mathcal{N}_{\rm R},
\end{aligned}\tag{P1}\label{P1}
\end{equation}
where $\mathrm{C}_1$--$\mathrm{C}_3$ are imposed for every symbol vector $\mathbf{s}$. In particular, $\mathrm{C}_1$ ensures that the observation of each Bob lies inside the CI region of its intended symbol, as in \eqref{eq:CIlin}, $\mathrm{C}_2$ ensures that the observation of the target/Eve lies inside the DI region \eqref{eq:DIunion}, $\mathrm{C}_3$ is the per-waveguide power constraint \eqref{eq:power}, and $\mathrm{C}_4$ confines each PA to its waveguide. By Lemma~\ref{th:factor}, the PEB of each symbol vector equals $c(\tilde{\bx}^{\rm R})/|\bh_e^H\bx|$, where the geometric term $c(\tilde{\bx}^{\rm R})$ depends only on the receive PA positions, which are common to all symbol vectors, whereas the illumination depends on the symbol vector through $\bx$. Hence, the objective of \eqref{P1} can be written as
\begin{equation}\label{eq:objsep}
\mathbb{E}_{\mathbf{s}}\left[\PEB\right]=c(\tilde{\bx}^{\rm R})\,\mathbb{E}_{\mathbf{s}}\left[\frac{1}{\left|\bh_e^H\bx\right|}\right],
\end{equation}
which yields the following exact decomposition of \eqref{P1} into three subproblems.

\begin{proposition}\label{prop:sep}
Problem \eqref{P1} is equivalent to the sequential solution of the following three problems. The first is the receive-placement problem
\begin{equation}
\min_{\tilde{\bx}^{\rm R}\in[0,D_x]^{N_{\rm R}}}\ c(\tilde{\bx}^{\rm R}),\tag{P1-R}\label{P1R}
\end{equation}
which is independent of the Bobs, the symbols, and the transmit side. The second is the per-symbol SLP problem, solved for every symbol vector and every transmit placement,
\begin{equation}
p^\star(\mathbf{s};\tilde{\bx}^{\rm T})=\max_{\bx}\ \left|\bh_e^H\bx\right|\quad\text{s.t.}\ \mathrm{C}_1,\ \mathrm{C}_2,\ \mathrm{C}_3.\tag{P1-S}\label{P1S}
\end{equation}
The third is the transmit-placement problem
\begin{equation}
\min_{\tilde{\bx}^{\rm T}\in[0,D_x]^{N_{\rm T}}}\ \mathbb{E}_{\mathbf{s}}\left[\frac{1}{p^\star(\mathbf{s};\tilde{\bx}^{\rm T})}\right].\tag{P1-T}\label{P1T}
\end{equation}
\end{proposition}
\begin{IEEEproof}
By \eqref{eq:objsep}, the objective of \eqref{P1} is the product of two positive factors that depend on disjoint sets of variables, and $\mathrm{C}_1$--$\mathrm{C}_3$ do not involve $\tilde{\bx}^{\rm R}$. Hence, the two factors are minimized separately, which yields \eqref{P1R}. For fixed $\tilde{\bx}^{\rm T}$, the second factor is a sum of terms, each depending on its own $\bx(\mathbf{s})$ under its own copy of $\mathrm{C}_1$--$\mathrm{C}_3$, so that it is minimized term by term, which yields \eqref{P1S} and, in turn, \eqref{P1T}.
\end{IEEEproof}

Proposition~\ref{prop:sep} contrasts with the joint placement and beamforming designs of the PASS literature, which couple the PA positions and the beamformers through the objective and update them by alternating optimization, in PASS communication \cite{Wang2025modeling}, PLS \cite{Papanikolaou2025AN}, and ISAC \cite{Li2025PASSISAC,Song2026} alike. Here, the factorization of Lemma~\ref{th:factor} separates the receive placement exactly and nests the per-symbol precoding inside the transmit placement, so that no iteration between the blocks is required and no optimality is lost.

\subsection{Receive Pinching-Antenna Placement}
By Proposition~\ref{prop:sep}, the receive PAs are placed by solving \eqref{P1R}. Since the receive channel enters $\mathbf{G}$ only through the magnitudes $|\ba_t|^{\circ 2}$ and the logarithmic derivatives $\mathbf{D}$, in which both the guided and the free-space phases of $\ba_t$ cancel by \eqref{eq:logder}, $c(\tilde{\bx}^{\rm R})$ is a smooth function without wavelength-scale oscillations. Problem \eqref{P1R} is therefore a smooth low-dimensional problem, which is nonconvex in general and admits no closed-form solution for an arbitrary layout. A closed-form characterization is obtained, however, for a symmetric layout of three receive waveguides, which is given by the following corollary of Lemma~\ref{lem:gdop}.

\begin{corollary}\label{lem:stagger}
Let $N_{\rm R}=3$ receive waveguides be located at $\tilde y^{\rm R}\in\{y_t-a,\,y_t,\,y_t+a\}$, let the PA of the middle waveguide be at $x_t$ and the PAs of the outer waveguides at $x_t+\delta$, let $r_\perp=\sqrt{a^2+d^2}$ be the perpendicular distance from the target to the outer waveguides, and let $\alpha=0$. Then, the information on $x_t$ is
\begin{equation}\label{eq:Gxx}
[\mathbf{G}]_{x_tx_t}=\frac{2\eta\kappa^2\,\delta^2}{\left(\delta^2+r_\perp^2\right)\left(\delta^2+r_\perp^2+2d^2\right)},
\end{equation}
which is maximized at $\delta=\delta^\star_{x}$, with
\begin{equation}\label{eq:dstar}
\delta^\star_{x}=\left[r_\perp^2\left(r_\perp^2+2d^2\right)\right]^{1/4}.
\end{equation}
Moreover, the PEB is minimized at $\delta=\delta^\star$, where $\delta^\star=\sqrt{t^\star}<r_\perp$ and $t^\star$ is the unique positive root of the cubic
\begin{equation}\label{eq:cubic}
\frac{2}{a^2}t^3+\left(1+\frac{2\rho}{a^2}\right)t^2-\rho\,(\rho+2d^2)=0,
\end{equation}
in which $t=\delta^2$ and $\rho=r_\perp^2$.
\end{corollary}
\begin{IEEEproof}
The proof is presented in Appendix~\ref{app:stagger}.
\end{IEEEproof}

Corollary~\ref{lem:stagger} should be contrasted with the communication-optimal placement of \cite{Xu2025rate,Tyrovolas2025}, which places each PA at zero offset from the served node. For sensing, this PA is the least informative about $x_t$, since a small displacement of the target along the waveguide changes its distance only to second order. The information on $x_t$ is maximized at the offset $\delta^\star_x>r_\perp$, whereas the information on $y_t$ is maximized at $\delta=0$, as shown by \eqref{eq:Gsym} in Appendix~\ref{app:stagger}, so that the PEB-optimal offset $\delta^\star<r_\perp$ is a compromise between the two.

For an arbitrary layout, \eqref{P1R} is solved by a cyclic one-dimensional search. Starting from an initial placement, the receive PAs are updated one at a time, with the remaining PAs fixed, as
\begin{equation}\label{eq:cyclic}
\tilde x^{\rm R}_m\leftarrow\arg\min_{\tilde x\in[0,D_x]}\ c\big(\tilde{\bx}^{\rm R}\big)\Big|_{\tilde x^{\rm R}_m=\tilde x},\qquad m\in\mathcal{N}_{\rm R},
\end{equation}
where the one-dimensional problem in \eqref{eq:cyclic} is solved by a global grid search over $[0,D_x]$ followed by a local refinement around the best grid point, which is reliable because $c(\cdot)$ is smooth. The cycles over $\mathcal{N}_{\rm R}$ are repeated until the relative decrease of $c(\tilde{\bx}^{\rm R})$ over a cycle falls below a tolerance. Since no update in \eqref{eq:cyclic} increases the objective, the sequence of objectives converges, and its limit is a coordinate-wise minimum of \eqref{P1R}, which depends on the initial placement.

The initial placement is obtained from Corollary~\ref{lem:stagger} as follows. Although the corollary is restricted to three waveguides, the offset \eqref{eq:dstar} depends on a waveguide only through its perpendicular distance from the target, which is defined for any receive waveguide as
\begin{equation}\label{eq:rperpm}
r_{\perp,m}=\sqrt{\left(\tilde y^{\rm R}_m-y_t\right)^2+d^2},\qquad m\in\mathcal{N}_{\rm R},
\end{equation}
whereas the PEB-optimal offset $\delta^\star$ depends also on the lateral spacing $a$ of the layout. Hence, \eqref{eq:dstar} is applied to every receive waveguide individually, and the initial placement is
\begin{equation}\label{eq:init}
\tilde x^{\rm R}_{m,0}=\Big[x_t+\sigma_m\left[r_{\perp,m}^2\left(r_{\perp,m}^2+2d^2\right)\right]^{1/4}\Big]_0^{D_x},\qquad m\in\mathcal{N}_{\rm R},
\end{equation}
where $[\cdot]_0^{D_x}$ denotes the projection onto $[0,D_x]$ and $\boldsymbol\sigma=[\sigma_1,\dots,\sigma_{N_{\rm R}}]^{\top}\in\{-1,+1\}^{N_{\rm R}}$ is a sign pattern, which determines the side of the target on which each PA lies. Since $c(\cdot)$ can have a local minimum for every sign pattern, and opposite signs on the outer waveguides of Corollary~\ref{lem:stagger} would even render the bearing vectors collinear and $x_t$ unobservable by Proposition~\ref{prop:ident}, the search \eqref{eq:cyclic} is run from all sign patterns and the best solution is retained, which is inexpensive because $c(\cdot)$ is evaluated in closed form and $N_{\rm R}$ is small.
\subsection{Symbol-Level Precoding for Fixed PA Positions}
By Proposition~\ref{prop:sep}, for given PA positions the transmit vector of every symbol interval is obtained from \eqref{P1S}, which is nonconvex, since $\mathrm{C}_2$ is a union of $M-1$ convex cones and the convex function $|\bh_e^H\bx|$ is maximized. Both difficulties are removed by treating each cone $\mathcal{D}(s')$ separately, since the cone confines the phase of the illumination to within $\Phi$ of the phase of $s'$, so that $|\bh_e^H\bx|$ can be replaced by its projection on the known direction of $s'$, which is linear in $\bx$. Unlike the secure SLP-ISAC designs of \cite{Su2022secureDFRC,Su2025BCRB,Bozanis2026robust}, in which the sensing objective is linearized around the current iteration and the resulting convex problem is solved repeatedly, no iterations are therefore needed, as shown in the following proposition.

\begin{proposition}\label{prop:exactDI}
For every wrong symbol $s'\ne s_1$, consider the convex problem
\begin{equation}
\begin{aligned}
\max_{\bx}\quad&\Re\{e^{-j\angle s'}\bh_e^H\bx\}\\
\text{s.t.}\quad&\mathrm{C}_1,\ \mathrm{C}_3,\ \bx\in\mathcal{D}(s').
\end{aligned}\tag{P2-$s'$}\label{P2}
\end{equation}
Let $\bx^\star_{s'}$ denote the solution of \eqref{P2}, and let $\hat{\bx}$ be the solution among $\bx^\star_{s'}$, $s'\ne s_1$, with the largest illumination $|\bh_e^H\bx^\star_{s'}|$. Then, $\hat{\bx}$ is feasible for \eqref{P1S}, and its illumination is within a worst-case factor $\cos\Phi$ of the global optimum of \eqref{P1S}, which tends to one as $M$ grows.
\end{proposition}
\begin{IEEEproof}
Since $\mathcal{D}_{\rm DI}$ is the union of the cones $\mathcal{D}(s')$, the maximum of $|\bh_e^H\bx|$ over the feasible set of \eqref{P1S} equals the largest of its maxima over the $M-1$ feasible sets of \eqref{P2}, so that treating each cone separately is exact. Inside $\mathcal{D}(s')$, the rotated observation $e^{-j\angle s'}p$ lies in a cone of half-angle $\Phi$ around the positive real axis, hence
\begin{equation}\label{eq:conebound}
|p|\cos\Phi\le\Re\left\{e^{-j\angle s'}p\right\}\le|p|.
\end{equation}
By the left inequality, the optimal value of \eqref{P2} is at least $\cos\Phi$ times the maximum of $|p|$ over the same feasible set, and, by the right inequality, the illumination of $\bx^\star_{s'}$ is at least the optimal value of \eqref{P2}. Taking the best of the $M-1$ solutions gives the bound with respect to the global optimum of \eqref{P1S}, which completes the proof.
\end{IEEEproof}

Each problem \eqref{P2} is a second-order cone program with a linear objective, $2K+2$ half-space constraints and $N_{\rm T}$ disc constraints, and can be solved by any convex solver. However, its structure is simple enough for a dedicated solution, which is worthwhile because \eqref{P2} is solved for every symbol vector and every candidate position. The following lemma provides its solution in closed form, up to the Lagrange multipliers of its half-space constraints, which are, in turn, obtained from a convex problem of dimension $2K+2$.

\begin{lemma}\label{lem:kkt}
Let $\bar{\ba}_i$ and $b_i$, $i=1,\dots,2K+2$, collect the vectors and thresholds of the CI and DI constraints of \eqref{P2}, i.e., $\tilde\bh_{k,i}$ and $\gamma_k$ in \eqref{eq:CIlin} for $i=1,\dots,2K$, and $\tilde\ba_i(s')$ and $\gamma_e$ in \eqref{eq:DI} for $i=2K+1,2K+2$, and let $\mathbf{c}=e^{j\angle s'}\bh_e$, so that the objective of \eqref{P2} is $\Re\{\mathbf{c}^H\bx\}$ and its CI and DI constraints are the half-spaces $\Re\{\bar{\ba}_i^H\bx\}\ge b_i$, $i=1,\dots,2K+2$. Let $\mu_i\ge0$ denote the Lagrange multiplier of the $i$-th half-space constraint, let $\boldsymbol\mu=[\mu_1,\dots,\mu_{2K+2}]^{\top}$, and let
\begin{equation}\label{eq:zeta}
\boldsymbol\zeta(\boldsymbol\mu)=\mathbf{c}+\sum_{i=1}^{2K+2}\mu_i\,\bar{\ba}_i.
\end{equation}
If \eqref{P2} is feasible and satisfies Slater's condition, the optimal multipliers $\boldsymbol\mu^\star$ are the solution of the convex problem
\begin{equation}\label{eq:dual}
\min_{\boldsymbol\mu\ge\mathbf{0}}\ \sqrt{\frac{P}{N_{\rm T}}}\sum_{n=1}^{N_{\rm T}}\left|\zeta_n(\boldsymbol\mu)\right|-\sum_{i=1}^{2K+2}\mu_ib_i,
\end{equation}
whose optimal value equals that of \eqref{P2}. Moreover, for every waveguide with $\zeta_n(\boldsymbol\mu^\star)\ne0$, the power constraint is active and the solution of \eqref{P2} is
\begin{equation}\label{eq:kkt}
x^\star_n=\sqrt{\frac{P}{N_{\rm T}}}\,\frac{\zeta_n(\boldsymbol\mu^\star)}{\left|\zeta_n(\boldsymbol\mu^\star)\right|},
\end{equation}
whereas, for every waveguide with $\zeta_n(\boldsymbol\mu^\star)=0$, the power constraint is not necessarily active and $x^\star_n$ is any primal-feasible value satisfying the equalities associated with the positive multipliers.
\end{lemma}
\begin{IEEEproof}
The proof is presented in Appendix~\ref{app:kkt}.
\end{IEEEproof}

By Lemma~\ref{lem:kkt}, the target/Eve enters \eqref{P2} in the same way as the Bobs, through the two half-spaces of $\mathcal{D}(s')$, i.e., it acts as a $(K+1)$-th virtual Bob with symbol $s'$ and SNR target $\Gamma_e$, so that \eqref{P2} inherits the closed-form structure of CI precoding for PSK \cite{Li2018closedform,Li2021multilevel}. In particular, the vector $\boldsymbol\zeta(\boldsymbol\mu^\star)$ in \eqref{eq:zeta} is the channel of the target corrected by the active CI and DI constraints, and, by \eqref{eq:kkt}, every waveguide with $\zeta_n(\boldsymbol\mu^\star)\ne0$ transmits at full power with the phase of this corrected channel. Moreover, every constraint of \eqref{P2} is a half-space or a disc, onto which the projection is available in closed form. Hence, in the numerical results, \eqref{P2} is solved by a consensus alternating direction method of multipliers (ADMM) \cite{Boyd2011ADMM} whose steps are these projections.



\subsection{Transmit Pinching-Antenna Placement}
The transmit placement is the only non-convex block that couples with the symbols: the guided phase of each waveguide is absorbed by $x_n$, but the free-space phases toward the Bobs and the target/Eve oscillate on the wavelength scale. We adopt the two-scale cyclic search that is standard for PASS \cite{Xu2025rate,Li2025PASSISAC}, with the difference that every candidate position is evaluated through the optimal value of the convex problem \eqref{P2}, so that feasibility of $\mathrm{C}_1$--$\mathrm{C}_3$ is preserved by construction. The transmit PAs are initialized with the communication-optimal rule of \cite{Xu2025rate,Tyrovolas2025}: the PA of the waveguide $\ell_k$ nearest to Bob $k$ is placed at
\begin{equation}\label{eq:txinit}
\tilde x^{\rm T}_{\ell_k}=x_k-\frac{\alpha\,r_{\perp,k}^2}{2},
\end{equation}
where $r_{\perp,k}$ is the perpendicular distance from Bob $k$ to that waveguide, and the remaining PAs are placed at the corresponding position of the target. If no subproblem \eqref{P2} is feasible for some symbol vector at this initialization, the waveguide midpoints and then the feasible layout of a $1$ m grid closest to \eqref{eq:txinit} are used instead. The objective of \eqref{P1T} is handled as follows. A common rotation of all symbols rotates the transmit vector by the same phase and leaves the illumination unchanged, so that the expectation over the $M^K$ symbol vectors reduces to the average $\mathcal{J}(\tilde{\bx}^{\rm T})$ of $1/p^\star(\mathbf{s};\tilde{\bx}^{\rm T})$ over the set $\mathcal{C}$ of the $M^{K-1}$ symbol vectors modulo a common rotation. Since $p^\star$ is the optimum of the nonconvex problem \eqref{P1S}, it is replaced by the illumination $\hat p(\mathbf{s};\tilde{\bx}^{\rm T})=|\bh_e^H\hat{\bx}|$ of the solution of Proposition~\ref{prop:exactDI}, which yields the surrogate objective
\begin{equation}\label{eq:J}
\hat{\mathcal{J}}(\tilde{\bx}^{\rm T})=\frac{1}{|\mathcal{C}|}\sum_{\mathbf{s}\in\mathcal{C}}\frac{1}{\hat p(\mathbf{s};\tilde{\bx}^{\rm T})}.
\end{equation}
Since $\cos\Phi\,p^\star\le\hat p\le p^\star$ by Proposition~\ref{prop:exactDI}, the surrogate satisfies $\mathcal{J}\le\hat{\mathcal{J}}\le\mathcal{J}/\cos\Phi$ for every transmit placement. Also, $\hat{\mathcal{J}}(\tilde{\bx}^{\rm T})=+\infty$ if, for some symbol vector, all $M-1$ subproblems \eqref{P2} are infeasible, and such placements are rejected, as summarized in Algorithm~\ref{alg:ao}.

\begin{algorithm}[t]
\caption{Proposed solution of \eqref{P1}}\label{alg:ao}
\begin{algorithmic}[1]
\Require Node positions, $P$, $\{\Gamma_k\}_{k\in\mathcal{K}}$, $\Gamma_e$, $\epsilon$
\Ensure $\tilde{\bx}^{\rm R}$, $\tilde{\bx}^{\rm T}$, $\hat{\bx}(\mathbf{s})$
\State Obtain $\tilde{\bx}^{\rm R}$ by \eqref{eq:cyclic} from every sign pattern of \eqref{eq:init}
\State Initialize $\tilde{\bx}^{\rm T}$ by \eqref{eq:txinit} or its feasible fallback
\Repeat
  \For{$n\in\mathcal{N}_{\rm T}$}
     \State Move $\tilde x^{\rm T}_n$ to the candidate with the smallest $\hat{\mathcal{J}}$ in \eqref{eq:J}, if it decreases $\hat{\mathcal{J}}$
  \EndFor
\Until{the relative decrease of $\hat{\mathcal{J}}$ is below $\epsilon$}
\State Compute $\hat{\bx}(\mathbf{s})$ by Proposition~\ref{prop:exactDI}
\end{algorithmic}
\end{algorithm}

\begin{proposition}\label{prop:conv}
Every placement accepted by Algorithm~\ref{alg:ao} is feasible for \eqref{P1} and strictly decreases $\hat{\mathcal{J}}$ in \eqref{eq:J}, so that the algorithm terminates after a finite number of cycles at a placement from which no move of a single transmit PA to a candidate position decreases $\hat{\mathcal{J}}$, i.e., at a coordinate-wise minimum of $\hat{\mathcal{J}}$ over the candidate positions. The complexity per cycle is
\begin{equation}\label{eq:complexity}
\mathcal{O}\left(N_{\rm T}\,N_{\rm c}\,|\mathcal{C}|\,(M-1)\,I_{\rm ADMM}\,K\,N_{\rm T}\right),
\end{equation}
where $N_{\rm c}$ is the number of candidate positions per PA and $I_{\rm ADMM}$ that of ADMM iterations.
\end{proposition}
\begin{IEEEproof}
A candidate position is accepted only if $\hat{\mathcal{J}}$ is finite and decreases. By \eqref{eq:J}, a finite $\hat{\mathcal{J}}$ means that, for every symbol vector, at least one of the $M-1$ subproblems \eqref{P2} is feasible, so that $\mathrm{C}_1$--$\mathrm{C}_3$ hold, whereas $\mathrm{C}_4$ holds since every candidate lies in $[0,D_x]$. Since $\hat{\mathcal{J}}$ decreases strictly over a finite candidate set, no placement is visited twice and the algorithm terminates. Finally, each cycle tests $N_{\rm c}$ candidates for each of the $N_{\rm T}$ PAs, each test solves \eqref{P2} for the $|\mathcal{C}|$ symbol classes and the $M-1$ wrong symbols, and each ADMM iteration costs $\mathcal{O}(KN_{\rm T})$ operations for the projections onto the $2K+2$ half-spaces, which gives \eqref{eq:complexity}.
\end{IEEEproof}

\section{Numerical Results}\label{sec:results}
This section evaluates the proposed design through Monte Carlo simulations. Unless otherwise stated, the carrier frequency is $f_c=28$ GHz, the effective refractive index is $n_{\rm eff}=1.4$, the in-waveguide attenuation is $\alpha=0.05$ Np/m, and $N_{\rm T}=N_{\rm R}=4$ waveguides are mounted at height $d=3$ m over a $10\times10$ m$^2$ service area \cite{Xu2025rate,Bozanis2025CRB}. The BS serves $K=2$ Bobs with QPSK and a transmit power of $P=30$ dBm, $\sigma_k^2=\sigma_e^2=\sigma_{\rm R}^2=-90$ dBm, $\Gamma_k=20$ dB, $\Gamma_e=0$ dB, and $\sigma_{\rm RCS}=1$ m$^2$. The results are obtained over $\mcNdrop$ random deployments, with the Bobs uniformly distributed over the service area and the target/Eve located within $3$ m of the protected Bob. A deployment is considered feasible if, for every symbol vector, at least one of the $M-1$ subproblems \eqref{P2} admits a solution, and the sensing performance is measured by the average PEB (APEB), i.e., the PEB averaged over the symbol vectors and then over the feasible deployments, shown only where at least half of the deployments are feasible.

The benchmarks are also optimized, employing the SLP design of Section~\ref{sec:design} with their own channels, so that the comparison isolates the gain of position reconfigurability. They include an unoptimized PASS, whose PAs are fixed at the waveguide midpoints to showcase the gain of the placement, and two conventional-array configurations placed at the center of the ceiling. The first employs a $4$-element $\lambda/2$-spaced ULA at both the transmit and receive sides. The second employs fully digital XL-MIMO arrays, each comprising $256$ $\lambda/2$-spaced elements and RF chains \cite{Cui2022}, i.e., $256$ RF chains per side against only $4$ for the PASS. All schemes use the same total transmit power and spherical-wave channel model, and the PEB of each scheme is evaluated under optimal combining. Algorithm~\ref{alg:ao} uses a receive grid of $0.05$ m, transmit candidates on a $0.5$ m grid and on a $\lambda/20$ grid within $\pm2\lambda$ of the current position, and $\epsilon=10^{-3}$.

Fig.~\ref{fig:cidi}(b) first verifies the CI/DI design, with $20$ noisy received symbols per receiver and deployment. The observations of the Bobs lie inside the CI region, whose margin is $\sqrt{\Gamma_k\sigma_k^2}=10\,\sigma_k$, whereas those of the target/Eve lie deep inside the wrong-symbol regions, with a median magnitude of $36\sqrt{\Gamma_k\sigma_k^2}$. Hence, the target/Eve decides a wrong symbol in every sample, despite the strong illumination required for sensing.

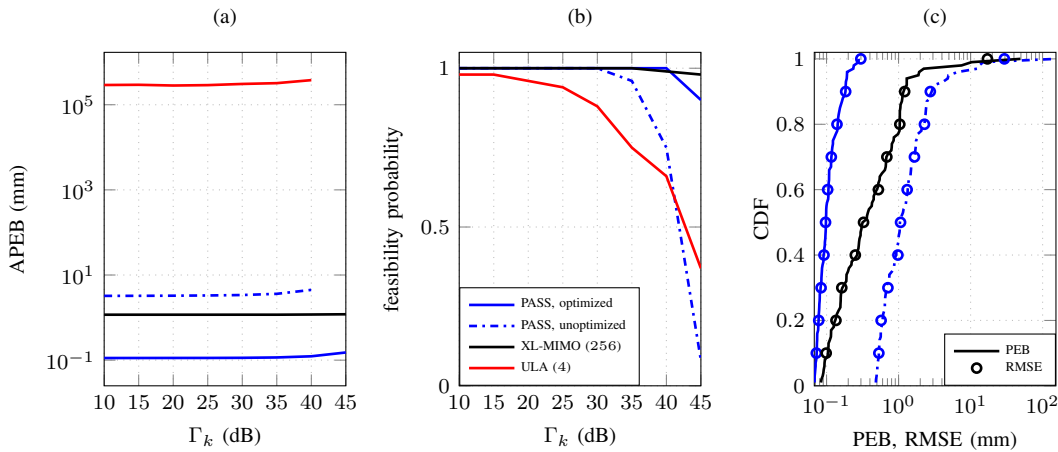
\begin{figure*}[t]\centering
\input{figures/figG.tex}
\caption{(a) APEB and (b) feasibility probability versus $\Gamma_k$, and (c) CDF of the PEB and RMSE.}\label{fig:bench}
\end{figure*}
Figs.~\ref{fig:bench}(a) and \ref{fig:bench}(b) compare the proposed design with the benchmarks in terms of APEB and feasibility. The compact ULA is practically unusable, with an APEB of about $300$ m, in agreement with Proposition~\ref{prop:compact}. The XL-MIMO array achieves an APEB of $1.17$ mm, while the optimized PASS reaches $0.11$ mm. This is $29$ times better than the unoptimized PASS and one order of magnitude better than XL-MIMO, with $64$ times fewer RF chains. PA placement also expands the feasible operating region. The optimized PASS is feasible in each deployment up to $40$ dB and in $90\%$ of the deployments at $45$ dB, compared with $75\%$ and $8\%$ for the unoptimized PASS. On average, the highest tested feasible $\Gamma_k$ of the optimized PASS is more than $5$ dB above that of the unoptimized PASS. The compact ULA is already infeasible in $12\%$ of the deployments at $30$ dB. The XL-MIMO array remains feasible in almost every deployment, owing to the beamforming gain of its $256$ transmit elements. The ULA, more than six orders of magnitude worse, is omitted from the remaining figures.

Fig.~\ref{fig:bench}(c) shows the distribution behind the APEB of Fig.~\ref{fig:bench}(a), together with the root-mean-square error (RMSE) $\sqrt{\mathbb{E}[\|\hat{\bpsi}_t-\bpsi_t\|^2]}$ of the maximum-likelihood estimate $\hat{\bpsi}_t$, i.e., the maximizer of $|\ba^H(\bpsi)\mathbf{y}_{\rm R}|^2/\|\ba(\bpsi)\|^2$, where the expectation is over the noise and the symbol vectors of each deployment. The illumination cancels in this ratio, so that the position is recovered from the receive geometry alone and the transmitter only sets the SNR, as Lemma~\ref{th:factor} and Remark~\ref{rem:txaperture} predict. The PEB of the optimized PASS lies between $0.073$ mm and $0.18$ mm in $80\%$ of the deployments and never exceeds $0.28$ mm, whereas the benchmarks are heavy-tailed, with an APEB three to four times their median. The RMSE follows the PEB within $10\%$ in each deployment of the optimized PASS and within $20\%$ for the benchmarks, except for one deployment of each benchmark with a PEB of several centimeters, where the RMSE falls below this local bound. Hence, the bound faithfully predicts the achievable accuracy.

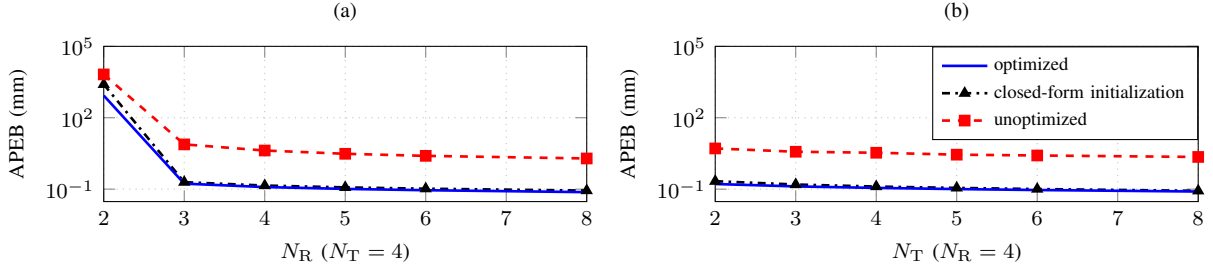
\begin{figure*}[t]\centering
\input{figures/figH.tex}
\caption{APEB versus (a) $N_{\rm R}$ and (b) $N_{\rm T}$.}\label{fig:wg}
\end{figure*}
Fig.~\ref{fig:wg} illustrates the distinct roles of the two apertures predicted by Lemma~\ref{th:factor}. Fig.~\ref{fig:wg}(a) examines the receive side. With $N_{\rm R}=2$, the phase term of $\mathbf{G}$ is rank deficient according to Proposition~\ref{prop:ident}, and position information is provided only by the amplitude variations discussed in Section~\ref{sec:geom}. Consequently, the APEB is $0.86$ m. Adding a third receive waveguide reduces it to $0.17$ mm, while further increasing $N_{\rm R}$ provides a more gradual improvement, reaching $0.075$ mm at $N_{\rm R}=8$. Since the transmit waveguides affect the PEB only through illumination, Fig.~\ref{fig:wg}(b) shows that increasing $N_{\rm T}$ from $2$ to $8$ reduces the APEB smoothly by a factor of two. In both cases, the unoptimized PASS is more than one order of magnitude worse. Moreover, the receive initialization of Corollary~\ref{lem:stagger} and the transmit initialization in \eqref{eq:txinit} remain within about $20\%$ of the respective converged solutions for $N_{\rm R}\ge3$ and $N_{\rm T}\ge3$, which supports the analytical placement rules.

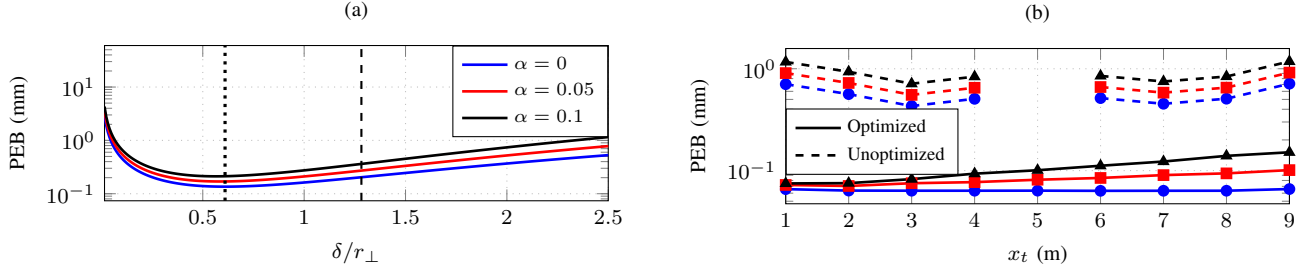
\begin{figure*}[t]\centering
{\setlength{\figw}{0.95\figwtwo}\setlength{\figh}{\fightwo}\input{figures/figD.tex}\hspace{0.03\textwidth}\input{figures/figJ.tex}}
\caption{PEB versus (a) $\delta/r_\perp$ and (b) $x_t$, for various attenuation coefficients.}\label{fig:att}
\end{figure*}
Fig.~\ref{fig:att} illustrates the geometric insights of Section~\ref{sec:metric} in a nominal geometry, with the Bobs located at $(6,3)$ m and $(3.5,7)$ m and the target/Eve at $(7.5,3.8)$ m. Fig.~\ref{fig:att}(a) validates Corollary~\ref{lem:stagger} for three receive waveguides at a spacing of $1.25$ m, the middle one passing above the target, with the waveguides taken long enough to accommodate every offset shown. As the receive PAs approach the points of their waveguides closest to the target, the PEB increases sharply and diverges as $\delta\to0$, being $20$ times larger than the optimum at the smallest plotted offset. The root of \eqref{eq:cubic} (dotted vertical line) coincides with the minimum of the lossless curve at $\delta\approx0.6\,r_\perp$, below the offset $\delta^\star_x$ of \eqref{eq:dstar} (dashed vertical line), and attenuation raises the curve and slightly shifts the optimum toward smaller offsets.

Fig.~\ref{fig:att}(b) shows the effect of the waveguide length and attenuation, obtained by moving the target/Eve and the protected Bob along the waveguides. Without attenuation, the PEB is flat at $0.063$ mm over the interior of the room and increases by about $4\%$ within one meter of either end, where the receive offsets of Corollary~\ref{lem:stagger} are truncated by the waveguide boundary. With attenuation, the PAs follow the target with the same relative positions, so that a target far from the feed forces both the transmit PAs that illuminate it and the receive PAs that collect its echo far along their waveguides. The in-waveguide losses are therefore paid twice, and the PEB grows with $x_t$ according to \eqref{eq:pebexp}. The unoptimized PASS is about eight times worse, and no value is shown for it at $x_t=5$ m, where all its receive PAs lie at $x_t$ and, by Proposition~\ref{prop:ident}, $x_t$ is unobservable, so that its PEB diverges as the target approaches the waveguide midpoint.

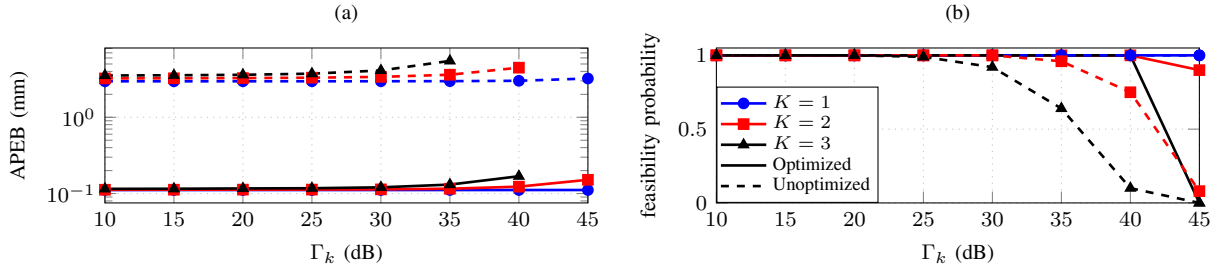
\begin{figure*}[t]\centering
\input{figures/figMC2.tex}
\caption{(a) APEB and (b) feasibility probability versus $\Gamma_k$, for various $K$.}\label{fig:mc2}
\end{figure*}
Fig.~\ref{fig:mc2} shows the communication--sensing trade-off for $K\in\{1,2,3\}$. With a single Bob, the design is feasible in every deployment for every threshold and the APEB is flat at $0.11$ mm, since one CI constraint and the DI constraint leave enough freedom to the transmit waveguides. With two and three Bobs, the optimized PASS is feasible in every deployment up to $40$ dB, and in $90\%$ of them at $45$ dB for $K=2$, whereas the unoptimized PASS loses feasibility more than $5$ dB earlier. Since $40$ dB lies far above the SNR required for reliable QPSK detection, the proposed design serves even three Bobs with a very demanding QoS, while its APEB remains below $0.17$ mm.

\section{Conclusion}\label{sec:conclusion}
In this paper, we studied secure near-field ISAC with PASS, where SLP guarantees CI for the Bobs and DI at the target/Eve. We showed that the Cartesian PEB factorizes into a transmit-side illumination term and a receive-side geometric term, which separates the joint design exactly into a receive-placement, a transmit-placement, and a per-symbol SLP problem, the latter handled through convex subproblems, and further yields closed-form placement rules for the receive pinching antennas. Numerical results confirmed sub-millimeter localization, erroneous decisions of the target/Eve in every simulated sample, and a wider range of feasible QoS thresholds than with fixed PAs. The value of PASS here lies in its high reconfigurability, which shapes the receive geometry so that this accuracy is achieved with only a few RF chains. Even a fully digital XL-MIMO array with $256$ RF chains per side remained one order of magnitude less accurate. Multiple pinching antennas per waveguide, imperfect knowledge of the target position, and placement driven by global bounds that account for ambiguities are left for future work.

\appendices
\section{Proof of Lemma~\ref{th:factor}}\label{app:factor}
Let $\mathbf{J}$ denote the FIM of $\boldsymbol\xi$. The PEB is the CRB on the root-mean-square error of the position estimate, i.e.,
\begin{equation}\label{eq:PEB}
\PEB=\sqrt{\tr\left(\left[\mathbf{J}^{-1}\right]_{1:2,1:2}\right)}.
\end{equation}
The observation \eqref{eq:echo} is complex Gaussian with covariance $\sigma_{\rm R}^2\mathbf{I}$ and mean $\boldsymbol\mu=\alpha_t\ba_t\bh_e^H\bx$, so that
\begin{equation}\label{eq:FIM}
\mathbf{J}=\frac{2}{\sigma_{\rm R}^2}\,\Re\left\{\dot{\boldsymbol\mu}^H\dot{\boldsymbol\mu}\right\},\qquad \dot{\boldsymbol\mu}=\frac{\partial\boldsymbol\mu}{\partial\boldsymbol\xi^{\top}}.
\end{equation}
Write $\boldsymbol\mu=\alpha_t\mathbf{v}(\bpsi_t)$ with $\mathbf{v}(\bpsi_t)=p\,\ba_t$ and $p=\bh_e^H\bx$, and let $\dot{\mathbf{V}}=\partial\mathbf{v}/\partial\bpsi_t^{\top}\in\mathbb{C}^{N_{\rm R}\times2}$. Since $\partial\boldsymbol\mu/\partial\bpsi_t^{\top}=\alpha_t\dot{\mathbf{V}}$, $\partial\boldsymbol\mu/\partial\Re\{\alpha_t\}=\mathbf{v}$, and $\partial\boldsymbol\mu/\partial\Im\{\alpha_t\}=j\mathbf{v}$, \eqref{eq:FIM} has the block form
\begin{equation}\label{eq:FIMblocks}
\mathbf{J}=\frac{2}{\sigma_{\rm R}^2}\begin{bmatrix}\mathbf{B} & \Re\{\mathbf{z}\} & -\Im\{\mathbf{z}\}\\ \Re\{\mathbf{z}\}^{\top} & \|\mathbf{v}\|^2 & 0\\ -\Im\{\mathbf{z}\}^{\top} & 0 & \|\mathbf{v}\|^2\end{bmatrix},
\end{equation}
where $\mathbf{B}=|\alpha_t|^2\Re\{\dot{\mathbf{V}}^H\dot{\mathbf{V}}\}$ and $\mathbf{z}=\alpha_t^*\dot{\mathbf{V}}^H\mathbf{v}\in\mathbb{C}^{2}$. By the block matrix inversion formula, $[\mathbf{J}^{-1}]_{1:2,1:2}=\mathbf{J}_{\rm e}^{-1}$, where $\mathbf{J}_{\rm e}$ is the Schur complement of the reflectivity block of \eqref{eq:FIMblocks}, which, since $\Re\{\mathbf{z}\}\Re\{\mathbf{z}\}^{\top}+\Im\{\mathbf{z}\}\Im\{\mathbf{z}\}^{\top}=\Re\{\mathbf{z}\mathbf{z}^H\}=|\alpha_t|^2\Re\{\dot{\mathbf{V}}^H\mathbf{v}\mathbf{v}^H\dot{\mathbf{V}}\}$, is given by
\begin{equation}\label{eq:schur}
\mathbf{J}_{\rm e}=\frac{2|\alpha_t|^2}{\sigma_{\rm R}^2}\,\Re\left\{\dot{\mathbf{V}}^H\boldsymbol\Pi^\perp_{\mathbf{v}}\dot{\mathbf{V}}\right\},\qquad
\boldsymbol\Pi^\perp_{\mathbf{v}}=\mathbf{I}-\frac{\mathbf{v}\mathbf{v}^H}{\|\mathbf{v}\|^2}.
\end{equation}
Since $p$ depends on $\bpsi_t$ through $\bh_e$, the derivative of $\mathbf{v}(\bpsi_t)$ is
\begin{equation}\label{eq:Vdot}
\dot{\mathbf{V}}=\ba_t\,\dot p+p\,\dot{\mathbf{A}},\qquad \dot{\mathbf{A}}=\frac{\partial\ba_t}{\partial\bpsi_t^{\top}}.
\end{equation}
Since $\mathbf{v}=p\,\ba_t$, the projector satisfies $\boldsymbol\Pi^\perp_{\mathbf{v}}=\boldsymbol\Pi^\perp_{\ba_t}$ and $\boldsymbol\Pi^\perp_{\ba_t}\ba_t=\mathbf{0}$, so that the transmit-side term $\ba_t\dot p$ in \eqref{eq:Vdot} vanishes and
\begin{equation}\label{eq:PiV}
\boldsymbol\Pi^\perp_{\mathbf{v}}\dot{\mathbf{V}}=p\,\boldsymbol\Pi^\perp_{\ba_t}\dot{\mathbf{A}}.
\end{equation}
Writing $\dot{\mathbf{A}}=\diag(\ba_t)\mathbf{D}$ with $\mathbf{D}$ as in \eqref{eq:D}, and using $\diag(\ba_t)^H\diag(\ba_t)=\diag(\mathbf{w})$, $\diag(\ba_t)^H\ba_t=\mathbf{w}$, and $\|\ba_t\|^2=\mathbf{1}^{\top}\mathbf{w}$, the projector satisfies
\begin{equation}\label{eq:Widentity}
\diag(\ba_t)^H\,\boldsymbol\Pi^\perp_{\ba_t}\,\diag(\ba_t)=\diag(\mathbf{w})-\frac{\mathbf{w}\mathbf{w}^{\top}}{\mathbf{1}^{\top}\mathbf{w}}=\mathbf{W},
\end{equation}
which is real. The substitution of \eqref{eq:PiV} into \eqref{eq:schur} then yields \eqref{eq:Je} with $\mathbf{G}$ as in \eqref{eq:G}. Finally, substituting $[\mathbf{J}^{-1}]_{1:2,1:2}=\mathbf{J}_{\rm e}^{-1}$ into \eqref{eq:PEB} gives \eqref{eq:PEBfactor}, which completes the proof.

\section{Proof of Lemma~\ref{lem:gdop}}\label{app:gdop}
From \eqref{eq:pa_channel}, the logarithmic derivative of the $m$-th receive channel with respect to the target coordinates is
\begin{equation}\label{eq:logder}
\frac{\partial\ln g(\bpsi_t;\tilde\bpsi^{\rm R}_m)}{\partial x_t}=\left(j\kappa+\frac{1}{r_m}\right)\frac{\tilde x^{\rm R}_m-x_t}{r_m},
\end{equation}
and similarly for $y_t$, so that the $m$-th row of $\mathbf{D}$ in \eqref{eq:D} is $(j\kappa+1/r_m)\,\mathbf{u}_m^{\top}$, with $\mathbf{u}_m$ as in Lemma~\ref{lem:gdop}. Hence,
\begin{equation}\label{eq:Ddecomp}
\mathbf{D}=j\kappa\,\mathbf{U}+\mathbf{E},\qquad \mathbf{U}=\left[\mathbf{u}_1,\dots,\mathbf{u}_{N_{\rm R}}\right]^{\top},
\end{equation}
where $\mathbf{E}=\diag(r_1^{-1},\dots,r_{N_{\rm R}}^{-1})\,\mathbf{U}$, and both $\mathbf{U}$ and $\mathbf{E}$ are real. Since $\mathbf{W}$ is also real, the cross terms of $\mathbf{D}^H\mathbf{W}\mathbf{D}$ are purely imaginary, so that
\begin{equation}\label{eq:ReDWD}
\Re\left\{\mathbf{D}^H\mathbf{W}\mathbf{D}\right\}=\kappa^2\,\mathbf{U}^{\top}\mathbf{W}\mathbf{U}+\mathbf{E}^{\top}\mathbf{W}\mathbf{E},
\end{equation}
where the second term is of relative order $(\kappa r_m)^{-2}$ with respect to the first. Moreover, since $\mathbf{U}^{\top}\diag(\mathbf{w})\mathbf{U}=\sum_{m}w_m\mathbf{u}_m\mathbf{u}_m^{\top}$ and $\mathbf{U}^{\top}\mathbf{w}=(\mathbf{1}^{\top}\mathbf{w})\,\bar{\mathbf{u}}$, substituting \eqref{eq:W} yields
\begin{equation}\label{eq:UWU}
\mathbf{U}^{\top}\mathbf{W}\mathbf{U}=\sum_{m=1}^{N_{\rm R}}w_m\left(\mathbf{u}_m-\bar{\mathbf{u}}\right)\left(\mathbf{u}_m-\bar{\mathbf{u}}\right)^{\top},
\end{equation}
which completes the proof.


\section{Proof of Corollary~\ref{lem:stagger}}\label{app:stagger}
With $r^2=\delta^2+r_\perp^2$, the weights and the bearing vectors of the three receive PAs are
\begin{align}
w_1&=w_3=\frac{\eta}{r^2},\qquad w_2=\frac{\eta}{d^2},\label{eq:symweights}\\
\mathbf{u}_1&=\frac{[\delta,-a]^{\top}}{r},\qquad \mathbf{u}_2=\mathbf{0},\qquad \mathbf{u}_3=\frac{[\delta,a]^{\top}}{r}.\label{eq:symbearings}
\end{align}
The weighted covariance \eqref{eq:gdop} is then diagonal, with
\begin{equation}\label{eq:Gsym}
[\mathbf{G}]_{x_tx_t}=\frac{2\kappa^2w_1w_2\,\delta^2}{(2w_1+w_2)\,r^2},\qquad [\mathbf{G}]_{y_ty_t}=\frac{2\kappa^2w_1a^2}{r^2}.
\end{equation}
Substituting the weights \eqref{eq:symweights} into \eqref{eq:Gsym} gives \eqref{eq:Gxx}, and setting its derivative with respect to $t=\delta^2$ to zero yields \eqref{eq:dstar}. Moreover, with $\rho=r_\perp^2$, the squared PEB is
\ifonecol
\begin{equation}\label{eq:PEB2sym}
\PEB^2\propto\frac{1}{[\mathbf{G}]_{x_tx_t}}+\frac{1}{[\mathbf{G}]_{y_ty_t}}=\frac{(t+\rho)(t+\rho+2d^2)}{2\eta\kappa^2\,t}+\frac{(t+\rho)^2}{2\eta\kappa^2a^2},
\end{equation}
\else
\begin{align}
\PEB^2&\propto\frac{1}{[\mathbf{G}]_{x_tx_t}}+\frac{1}{[\mathbf{G}]_{y_ty_t}}\nonumber\\
&=\frac{(t+\rho)(t+\rho+2d^2)}{2\eta\kappa^2\,t}+\frac{(t+\rho)^2}{2\eta\kappa^2a^2},\label{eq:PEB2sym}
\end{align}
\fi
whose derivative with respect to $t$ equals $q(t)/(2\eta\kappa^2t^2)$, where $q(t)$ denotes the polynomial on the left-hand side of \eqref{eq:cubic}. Since $q(t)$ is strictly increasing for $t>0$, with $q(0)=-\rho(\rho+2d^2)<0$ and $q(t)\to\infty$ as $t\to\infty$, it has exactly one positive root $t^\star$. Moreover, $q(t)<0$ for $0<t<t^\star$ and $q(t)>0$ for $t>t^\star$, so that the PEB is decreasing on $(0,t^\star)$ and increasing on $(t^\star,\infty)$. Hence, $t^\star$ is the unique minimizer. Finally, $q(\rho)=2\rho\left(2\rho^2/a^2-d^2\right)>0$, since $\rho\ge a^2$ and $\rho>d^2$ imply $\rho^2/a^2>d^2$, so that $t^\star<\rho$, i.e., $\delta^\star<r_\perp$, which completes the proof.


\section{Proof of Lemma~\ref{lem:kkt}}\label{app:kkt}
Problem \eqref{P2} is convex, with a linear objective, affine half-space constraints, and convex disc constraints, so that, under the Slater condition assumed in Lemma~\ref{lem:kkt}, the Karush--Kuhn--Tucker (KKT) conditions are necessary and sufficient for optimality \cite{Boyd2004}. Let $\nu_n\ge0$ denote the Lagrange multiplier of the $n$-th power constraint and $\boldsymbol\nu=[\nu_1,\dots,\nu_{N_{\rm T}}]^{\top}$. Writing \eqref{P2} as the minimization of $-\Re\{\mathbf{c}^H\bx\}$, the Lagrangian is
\ifonecol
\begin{align}
\mathcal{L}&=-\Re\{\mathbf{c}^H\bx\}+\sum_{i=1}^{2K+2}\mu_i\left(b_i-\Re\{\bar{\ba}_i^H\bx\}\right)+\sum_{n=1}^{N_{\rm T}}\nu_n\left(|x_n|^2-\frac{P}{N_{\rm T}}\right)\nonumber\\
&=\sum_{i=1}^{2K+2}\mu_ib_i-\sum_{n=1}^{N_{\rm T}}\left(\Re\{\zeta_n^*x_n\}-\nu_n|x_n|^2\right)-\frac{P}{N_{\rm T}}\sum_{n=1}^{N_{\rm T}}\nu_n,\label{eq:lagr}
\end{align}
\else
\begin{align}
\mathcal{L}&=-\Re\{\mathbf{c}^H\bx\}+\sum_{i=1}^{2K+2}\mu_i\left(b_i-\Re\{\bar{\ba}_i^H\bx\}\right)\nonumber\\
&\quad+\sum_{n=1}^{N_{\rm T}}\nu_n\left(|x_n|^2-\frac{P}{N_{\rm T}}\right)\nonumber\\
&=\sum_{i=1}^{2K+2}\mu_ib_i-\sum_{n=1}^{N_{\rm T}}\left(\Re\{\zeta_n^*x_n\}-\nu_n|x_n|^2\right)\nonumber\\
&\quad-\frac{P}{N_{\rm T}}\sum_{n=1}^{N_{\rm T}}\nu_n,\label{eq:lagr}
\end{align}
\fi
where the second equality follows from \eqref{eq:zeta}. Since \eqref{eq:lagr} is separable in the entries of $\bx$, the stationarity condition $\partial\mathcal{L}/\partial x_n^*=0$ yields
\begin{equation}\label{eq:stationarity}
2\nu_nx_n=\zeta_n,\qquad \forall n\in\mathcal{N}_{\rm T},
\end{equation}
whereas the complementary slackness conditions are
\begin{equation}\label{eq:cs}
\mu_i\left(b_i-\Re\{\bar{\ba}_i^H\bx\}\right)=0,\qquad \nu_n\left(|x_n|^2-\frac{P}{N_{\rm T}}\right)=0,
\end{equation}
i.e., a multiplier is positive only if the corresponding constraint holds with equality.

In the first case, $\nu_n>0$, so that the $n$-th power constraint holds with equality by \eqref{eq:cs}, and \eqref{eq:stationarity} gives
\begin{equation}\label{eq:xn}
x_n=\frac{\zeta_n}{2\nu_n}.
\end{equation}
Taking the modulus of \eqref{eq:xn} and using the equality in the power constraint yields
\begin{equation}\label{eq:nustar}
\nu_n=\frac{|\zeta_n|}{2\sqrt{P/N_{\rm T}}},
\end{equation}
and the substitution of \eqref{eq:nustar} into \eqref{eq:xn} gives \eqref{eq:kkt}. Conversely, $\zeta_n\ne0$ implies $\nu_n>0$ by \eqref{eq:stationarity}, so that \eqref{eq:kkt} holds for every waveguide with $\zeta_n\ne0$.

In the second case, $\nu_n=0$ and \eqref{eq:stationarity} reduces to $\zeta_n=0$, so that $x_n$ is fixed only by the feasibility of \eqref{P2} and the equalities in \eqref{eq:cs}.

It remains to show that $\boldsymbol\mu^\star$ solves \eqref{eq:dual}. By \eqref{eq:lagr}, for $\nu_n>0$ the infimum of $\mathcal{L}$ over $x_n$ is attained at \eqref{eq:xn} and equals $-|\zeta_n|^2/(4\nu_n)-\nu_nP/N_{\rm T}$, whereas for $\nu_n=0$ it equals $0$ if $\zeta_n=0$ and $-\infty$ otherwise. Hence, the dual function is
\begin{equation}\label{eq:dualfun}
g(\boldsymbol\mu,\boldsymbol\nu)=\sum_{i=1}^{2K+2}\mu_ib_i-\sum_{n=1}^{N_{\rm T}}\left(\frac{|\zeta_n|^2}{4\nu_n}+\nu_n\frac{P}{N_{\rm T}}\right)
\end{equation}
on its domain, and the dual problem is the maximization of \eqref{eq:dualfun} over $\boldsymbol\mu\ge\mathbf{0}$ and $\boldsymbol\nu\ge\mathbf{0}$. The maximization over $\boldsymbol\nu$ is carried out in closed form, since the $n$-th term of the second sum in \eqref{eq:dualfun} is minimized at the value \eqref{eq:nustar} of $\nu_n$, with minimum $\sqrt{P/N_{\rm T}}\,|\zeta_n|$, which also covers the case $\zeta_n=0$. Substituting into \eqref{eq:dualfun} and reversing the sign turns the maximization over $\boldsymbol\mu$ into the minimization \eqref{eq:dual}, whose objective is convex, as the sum of the moduli of affine functions of $\boldsymbol\mu$ and of a linear term. Finally, since \eqref{P2} is convex, strong duality holds under Slater's condition \cite{Boyd2004}, so that the optimal value of \eqref{eq:dual} equals that of \eqref{P2}, and $\bx^\star$ is the point that satisfies the KKT conditions with $\boldsymbol\mu^\star$, which completes the proof.

\bibliographystyle{IEEEtran}
\bibliography{Bibliography}
\end{document}

%% file: figures/data/values.tex
\def\fbAa{45}
\def\fbAb{45}
\def\fbAc{37.5}
\def\fbBa{35}
\def\fbBb{40}
\def\fbBc{45}

\def\dstarJ{1.2824}
\def\dstarPEB{0.6099}
\def\mcNdrop{100}

%% file: figures/data/values_mc.tex
\def\mcNdrop{100}
\def\ciMargin{10}
\def\diMargin{1}
\def\scMaxAbs{430.34}

%% file: figures/figCIDI.tex
\begin{tikzpicture}[scale={0.0041*\figwtwo}, font=\footnotesize, baseline=(current bounding box.center)]
\node[above] at (0,3.4) {(a)};
\clip (-3.8,-3.4) rectangle (3.8,3.4);
\colorlet{ciFill}{blue!15}
\colorlet{diFill}{red!15}
\fill[diFill] (-0.8,0) -- (-4.2,3.4) -- (-4.2,-3.4) -- cycle;   
\fill[diFill] (0,0.8) -- (2.6,3.4) -- (-2.6,3.4) -- cycle;      
\fill[diFill] (0,-0.8) -- (2.6,-3.4) -- (-2.6,-3.4) -- cycle;   
\fill[ciFill] (1.6,0) -- (4.2,2.6) -- (4.2,-2.6) -- cycle;
\draw[gray!70, dotted, thick] (-3.4,-3.4) -- (3.4,3.4);
\draw[gray!70, dotted, thick] (-3.4,3.4) -- (3.4,-3.4);
\draw[red!60!black, densely dashed, thick] (-4.2,3.4) -- (-0.8,0) -- (-4.2,-3.4);
\draw[red!60!black, densely dashed, thick] (2.6,3.4) -- (0,0.8) -- (-2.6,3.4);
\draw[red!60!black, densely dashed, thick] (2.6,-3.4) -- (0,-0.8) -- (-2.6,-3.4);
\draw[blue!60!black, densely dashed, thick] (4.2,2.6) -- (1.6,0) -- (4.2,-2.6);
\draw[->, thick] (-3.8,0) -- (3.7,0) node[above left] {Re};
\draw[->, thick] (0,-3.4) -- (0,3.3) node[below right] {Im};
\draw[blue!60!black] (2.3,0) arc (0:45:0.7);
\node[blue!60!black] at (2.5,0.36) {$\Phi$};
\fill (0,0) circle (1.5pt) node[below left, inner sep=2pt] {$0$};
\fill (1.6,0) circle (1.5pt);
\node[anchor=north east, inner sep=1pt] at (1.75,-0.12) {$\sqrt{\Gamma_k\sigma_k^2}$};
\fill (-0.8,0) circle (1.5pt);
\fill (0,0.8) circle (1.5pt);
\fill (0,-0.8) circle (1.5pt);
\node[anchor=south east, inner sep=1pt, font=\scriptsize] at (-1.25,0.08) {$\sqrt{\Gamma_e\sigma_e^2}$};
\node[blue!40!black, anchor=east] at (3.75,1.0) {CI region};
\node[red!50!black] at (-2.7,-1.1) {$\mathcal{D}(-s_1)$};
\node[red!50!black] at (0.85,2.65) {$\mathcal{D}(s_1e^{j\pi/2})$};
\node[red!50!black] at (-0.85,-2.7) {$\mathcal{D}(s_1e^{-j\pi/2})$};
\fill[blue!70!black] (3.0,-0.7) circle (1.8pt) node[right, inner sep=2pt] {$\tilde y_k$};
\fill[red!70!black] (-3.2,1.0) circle (1.8pt) node[above, inner sep=2pt] {$\tilde y_e$};
\end{tikzpicture}

%% file: figures/figC.tex
\pgfmathsetmacro{\scL}{1.08*\scMaxAbs}
\pgfmathsetmacro{\zmL}{2.5*\ciMargin}
\newcommand{\cidiregions}[1]{%
  \addplot[fill=red!15, draw=red!60!black, densely dashed, thick, forget plot] coordinates
    {(-\diMargin,0) (-#1,#1-\diMargin) (-#1,-#1+\diMargin)} -- cycle;
  \addplot[fill=red!15, draw=red!60!black, densely dashed, thick, forget plot] coordinates
    {(0,\diMargin) (#1-\diMargin,#1) (-#1+\diMargin,#1)} -- cycle;
  \addplot[fill=red!15, draw=red!60!black, densely dashed, thick, forget plot] coordinates
    {(0,-\diMargin) (#1-\diMargin,-#1) (-#1+\diMargin,-#1)} -- cycle;
  \addplot[fill=blue!15, draw=blue!60!black, densely dashed, thick, forget plot] coordinates
    {(\ciMargin,0) (#1,#1-\ciMargin) (#1,-#1+\ciMargin)} -- cycle;
  \addplot[gray!70, dotted, thick, no markers, forget plot] coordinates {(-#1,-#1) (#1,#1)};
  \addplot[gray!70, dotted, thick, no markers, forget plot] coordinates {(-#1,#1) (#1,-#1)};}
\begin{tikzpicture}[baseline=(current bounding box.center)]
\begin{axis}[cidi, name=main, axis equal image, xlabel={In-phase (rotated)}, ylabel={Quadrature (rotated)},
  xmin=-\scL, xmax=\scL, ymin=-\scL, ymax=\scL, width=\figwtwo, height=\figwtwo, title={(b)},
  grid=none, legend image post style={mark size=1.6pt},
  legend style={at={(0,0)}, anchor=south west, fill=white, fill opacity=0.85, draw opacity=1, text opacity=1}]
\cidiregions{\scL}
\addplot[only marks, mark=*, mark size=0.6pt, blue] table[x=I, y=Q] {\datapath figC_u1.dat}; \addlegendentry{Bob $1$}
\addplot[only marks, mark=*, mark size=0.6pt, black] table[x=I, y=Q] {\datapath figC_u2.dat}; \addlegendentry{Bob $2$}
\addplot[only marks, mark=*, mark size=0.6pt, red] table[x=I, y=Q] {\datapath figC_eve.dat}; \addlegendentry{target/Eve}
\node[blue!40!black, font=\footnotesize] at (axis cs:0.78*\scL,0.5*\scL) {CI region};
\node[red!50!black, font=\footnotesize] at (axis cs:-0.7*\scL,-0.4*\scL) {$\mathcal{D}(-s_1)$};
\node[red!50!black, font=\footnotesize] at (axis cs:0.45*\scL,0.82*\scL) {$\mathcal{D}(s_1e^{j\pi/2})$};
\draw[black, line width=0.5pt] (axis cs:-\zmL,-\zmL) rectangle (axis cs:\zmL,\zmL);
\end{axis}
\begin{axis}[at={(main.north west)}, anchor=north west, scale only axis,
  width=0.3\figwtwo, height=0.3\figwtwo, xmin=-\zmL, xmax=\zmL, ymin=-\zmL, ymax=\zmL,
  xtick=\empty, ytick=\empty,
  axis background/.style={fill=white}, enlarge x limits=false]
\cidiregions{\zmL}
\addplot[only marks, mark=*, mark size=0.6pt, blue, forget plot] table[x=I, y=Q] {\datapath figC_u1.dat};
\addplot[only marks, mark=*, mark size=0.6pt, black, forget plot] table[x=I, y=Q] {\datapath figC_u2.dat};
\addplot[only marks, mark=*, mark size=0.6pt, red, forget plot] table[x=I, y=Q] {\datapath figC_eve.dat};
\end{axis}
\end{tikzpicture}

%% file: figures/figG.tex
\begin{tikzpicture}
\begin{groupplot}[group style={group size=3 by 1, horizontal sep=1.5cm}, cidi, width=0.85\figwthree, height=\fighthree]
\nextgroupplot[ymode=log, ylabel={APEB (mm)}, title={(a)}, xlabel={$\Gamma_k$ (dB)}, xmin=10, xmax=45, xtick={10,15,...,45}]
\addplot[c1, no marks] table[x=gam, y=pass_opt] {\datapath figG.dat};
\addplot[c1, dashdotted, no marks] table[x=gam, y=pass_unopt] {\datapath figG.dat};
\addplot[c3, no marks] table[x=gam, y=xl] {\datapath figG.dat};
\addplot[c2, no marks] table[x=gam, y=ula] {\datapath figG.dat};
\nextgroupplot[ylabel={feasibility probability}, ymin=0, ymax=1.05, title={(b)}, xlabel={$\Gamma_k$ (dB)}, xmin=10, xmax=45, xtick={10,15,...,45},
  legend style={at={(0,0)}, anchor=south west, font=\tiny, row sep=-1pt}]
\addplot[c1, no marks] table[x=gam, y=p_opt] {\datapath figG.dat}; \addlegendentry{PASS, optimized}
\addplot[c1, dashdotted, no marks] table[x=gam, y=p_unopt] {\datapath figG.dat}; \addlegendentry{PASS, unoptimized}
\addplot[c3, no marks] table[x=gam, y=p_xl] {\datapath figG.dat}; \addlegendentry{XL-MIMO ($256$)}
\addplot[c2, no marks] table[x=gam, y=p_ula] {\datapath figG.dat}; \addlegendentry{ULA ($4$)}
\nextgroupplot[xmode=log, xlabel={PEB, RMSE (mm)}, ylabel={CDF}, ylabel shift=-4pt, ymin=0, ymax=1.02, title={(c)},
  legend style={at={(1,0)}, anchor=south east, font=\tiny, row sep=-2.5pt}]
\addlegendimage{black, no marks} \addlegendentry{PEB}
\addlegendimage{black, only marks, mark=o} \addlegendentry{RMSE}
\addplot[c1, no marks] table[x=peb, y=cdf] {\datapath figMSEcdf_opt.dat};
\addplot[c1, only marks, mark=o, mark repeat=10, mark phase=10] table[x=rmse, y=cdf] {\datapath figMSEcdf_opt.dat};
\addplot[c1, dashdotted, no marks] table[x=peb, y=cdf] {\datapath figMSEcdf_unopt.dat};
\addplot[c1, only marks, mark=o, mark repeat=10, mark phase=10] table[x=rmse, y=cdf] {\datapath figMSEcdf_unopt.dat};
\addplot[c3, no marks] table[x=peb, y=cdf] {\datapath figMSEcdf_xl.dat};
\addplot[c3, only marks, mark=o, mark repeat=10, mark phase=10] table[x=rmse, y=cdf] {\datapath figMSEcdf_xl.dat};
\end{groupplot}
\end{tikzpicture}

%% file: figures/figH.tex
\begin{tikzpicture}
\begin{groupplot}[group style={group size=2 by 1, horizontal sep=1.7cm}, cidi, width=0.9\figwtwo, height=\fightwo,
  ymode=log, ylabel={APEB (mm)}, ymin=0.03, ymax=1e5, xtick={2,3,...,8}]
\nextgroupplot[xlabel={$N_{\rm R}$ ($N_{\rm T}=4$)}, title={(a)}]
\addplot[c3, dashdotted] table[x=MR, y=init] {\datapath figHa.dat};
\addplot[c2, dashed] table[x=MR, y=unopt] {\datapath figHa.dat};
\addplot[blue, no markers] table[x=MR, y=opt] {\datapath figHa.dat};
\nextgroupplot[xlabel={$N_{\rm T}$ ($N_{\rm R}=4$)}, title={(b)}, legend style={at={(1,1)}, anchor=north east}]
\addlegendimage{blue, no markers, line width=1pt} \addlegendentry{optimized}
\addplot[c3, dashdotted] table[x=MT, y=init] {\datapath figHb.dat}; \addlegendentry{closed-form initialization}
\addplot[c2, dashed] table[x=MT, y=unopt] {\datapath figHb.dat}; \addlegendentry{unoptimized}
\addplot[blue, no markers, forget plot] table[x=MT, y=opt] {\datapath figHb.dat};
\end{groupplot}
\end{tikzpicture}

%% file: figures/figD.tex
\begin{tikzpicture}
\begin{semilogyaxis}[cidi, xlabel={$\delta/r_\perp$}, ylabel={PEB (mm)}, no markers, title={(a)}, ymax=60,
  legend style={at={(1,1)}, anchor=north east}]
\addplot[c1] table[x=x, y=peb0] {\datapath figD.dat}; \addlegendentry{$\alpha=0$}
\addplot[c2] table[x=x, y=peb005] {\datapath figD.dat}; \addlegendentry{$\alpha=0.05$}
\addplot[c3] table[x=x, y=peb01] {\datapath figD.dat}; \addlegendentry{$\alpha=0.1$}
\draw[black, dashed, line width=0.8pt] ({axis cs:\dstarJ,1}|-{rel axis cs:0,0}) -- ({axis cs:\dstarJ,1}|-{rel axis cs:0,1});
\draw[black, dotted, line width=1.2pt] ({axis cs:\dstarPEB,1}|-{rel axis cs:0,0}) -- ({axis cs:\dstarPEB,1}|-{rel axis cs:0,1});
\end{semilogyaxis}
\end{tikzpicture}

%% file: figures/figJ.tex
\begin{tikzpicture}
\begin{semilogyaxis}[cidi, xlabel={$x_t$ (m)}, ylabel={PEB (mm)}, xmin=1, xmax=9, xtick={1,2,...,9}, title={(b)},
  legend style={at={(0,0.4)}, anchor=west}]
\addlegendimage{black, solid, no markers}\addlegendentry{Optimized}
\addlegendimage{black, dashed, no markers}\addlegendentry{Unoptimized}
\addplot[c1, forget plot] table[x=xt, y=opt0] {\datapath figJ.dat};
\addplot[c2, forget plot] table[x=xt, y=opt005] {\datapath figJ.dat};
\addplot[c3, forget plot] table[x=xt, y=opt01] {\datapath figJ.dat};
\addplot[c1, dashed, forget plot] table[x=xt, y=un0] {\datapath figJ.dat};
\addplot[c2, dashed, forget plot] table[x=xt, y=un005] {\datapath figJ.dat};
\addplot[c3, dashed, forget plot] table[x=xt, y=un01] {\datapath figJ.dat};
\end{semilogyaxis}
\end{tikzpicture}

%% file: figures/figMC2.tex
\begin{tikzpicture}
\begin{groupplot}[group style={group size=2 by 1, horizontal sep=1.7cm}, cidi, width=0.9\figwtwo, height=\fightwo,
  xlabel={$\Gamma_k$ (dB)}, xmin=10, xmax=45, xtick={10,15,...,45}]
\nextgroupplot[ymode=log, ylabel={APEB (mm)}, title={(a)}]
\addplot[c1] table[x=gam, y=med_opt1] {\datapath figMC2.dat};
\addplot[c2] table[x=gam, y=med_opt2] {\datapath figMC2.dat};
\addplot[c3] table[x=gam, y=med_opt3] {\datapath figMC2.dat};
\addplot[c1, dashed] table[x=gam, y=med_un1] {\datapath figMC2.dat};
\addplot[c2, dashed] table[x=gam, y=med_un2] {\datapath figMC2.dat};
\addplot[c3, dashed] table[x=gam, y=med_un3] {\datapath figMC2.dat};
\nextgroupplot[ylabel={feasibility probability}, ymin=0, ymax=1.05, title={(b)},
  legend style={at={(0,0)}, anchor=south west, row sep=-2.5pt, inner sep=1pt}]
\addplot[c1] table[x=gam, y=p_opt1] {\datapath figMC2.dat}; \addlegendentry{$K=1$}
\addplot[c2] table[x=gam, y=p_opt2] {\datapath figMC2.dat}; \addlegendentry{$K=2$}
\addplot[c3] table[x=gam, y=p_opt3] {\datapath figMC2.dat}; \addlegendentry{$K=3$}
\addlegendimage{black, solid, no markers}\addlegendentry{Optimized}
\addlegendimage{black, dashed, no markers}\addlegendentry{Unoptimized}
\addplot[c1, dashed, forget plot] table[x=gam, y=p_un1] {\datapath figMC2.dat};
\addplot[c2, dashed, forget plot] table[x=gam, y=p_un2] {\datapath figMC2.dat};
\addplot[c3, dashed, forget plot] table[x=gam, y=p_un3] {\datapath figMC2.dat};
\end{groupplot}
\end{tikzpicture}

%% file: Bibliography.bib
@article{Liu2022ISAC,
  author={Liu, Fan and Cui, Yuanhao and Masouros, Christos and Xu, Jie and Han, Tony Xiao and Eldar, Yonina C. and Buzzi, Stefano},
  title={{Integrated Sensing and Communications: Toward Dual-Functional Wireless Networks for {6G} and Beyond}},
  journal={IEEE J. Sel. Areas Commun.}, volume={40}, number={6}, pages={1728--1767}, year={2022}}

@article{Liu2018MUMIMO,
  author={Liu, Fan and Masouros, Christos and Li, Ang and Sun, Huafei and Hanzo, Lajos},
  title={{{MU-MIMO} Communications With {MIMO} Radar: From Co-Existence to Joint Transmission}},
  journal={IEEE Trans. Wireless Commun.}, volume={17}, number={4}, pages={2755--2770}, year={2018}}

@article{Liu2022CRB,
  author={Liu, Fan and Liu, Ya-Feng and Li, Ang and Masouros, Christos and Eldar, Yonina C.},
  title={{Cram\'er-{R}ao Bound Optimization for Joint Radar-Communication Beamforming}},
  journal={IEEE Trans. Signal Process.}, volume={70}, pages={240--253}, year={2022}}

@article{Wei2022security,
  author={Wei, Zhongxiang and Liu, Fan and Masouros, Christos and Su, Nanchi and Petropulu, Athina P.},
  title={{Toward Multi-Functional {6G} Wireless Networks: Integrating Sensing, Communication, and Security}},
  journal={IEEE Commun. Mag.}, volume={60}, number={4}, pages={65--71}, year={2022}}

@article{Su2021malicious,
  author={Su, Nanchi and Liu, Fan and Masouros, Christos},
  title={{Secure Radar-Communication Systems With Malicious Targets: Integrating Radar, Communications and Jamming Functionalities}},
  journal={IEEE Trans. Wireless Commun.}, volume={20}, number={1}, pages={83--95}, year={2021}}

@article{Su2024sensingassisted,
  author={Su, Nanchi and Liu, Fan and Masouros, Christos},
  title={{Sensing-Assisted Eavesdropper Estimation: An {ISAC} Breakthrough in Physical Layer Security}},
  journal={IEEE Trans. Wireless Commun.}, volume={23}, number={4}, pages={3162--3174}, year={2024}}

@article{Masouros2011,
  author={Masouros, Christos},
  title={{Correlation Rotation Linear Precoding for {MIMO} Broadcast Communications}},
  journal={IEEE Trans. Signal Process.}, volume={59}, number={1}, pages={252--262}, year={2011}}

@article{Li2020tutorial,
  author={Li, Ang and Spano, Danilo and Krivochiza, Jevgenij and Domouchtsidis, Stavros and Tsinos, Christos G. and Masouros, Christos and Chatzinotas, Symeon and Li, Yonghui and Vucetic, Branka and Ottersten, Bj\"orn},
  title={{A Tutorial on Interference Exploitation via Symbol-Level Precoding: Overview, State-of-the-Art and Future Directions}},
  journal={IEEE Commun. Surveys Tuts.}, volume={22}, number={2}, pages={796--839}, year={2020}}

@article{Li2018closedform,
  author={Li, Ang and Masouros, Christos},
  title={{Interference Exploitation Precoding Made Practical: Optimal Closed-Form Solutions for {PSK} Modulations}},
  journal={IEEE Trans. Wireless Commun.}, volume={17}, number={11}, pages={7661--7676}, year={2018}}

@article{Li2021multilevel,
  author={Li, Ang and Masouros, Christos and Vucetic, Branka and Li, Yonghui and Swindlehurst, A. Lee},
  title={{Interference Exploitation Precoding for Multi-Level Modulations: Closed-Form Solutions}},
  journal={IEEE Trans. Commun.}, volume={69}, number={1}, pages={291--308}, year={2021}}

@article{Khandaker2018,
  author={Khandaker, Muhammad R. A. and Masouros, Christos and Wong, Kai-Kit},
  title={{Constructive Interference Based Secure Precoding: A New Dimension in Physical Layer Security}},
  journal={IEEE Trans. Inf. Forensics Security}, volume={13}, number={9}, pages={2256--2268}, year={2018}}

@article{Wang2025uniform,
  author={Wang, Yiran and Hu, Xiaoyan and Li, Ang and Masouros, Christos and Wong, Kai-Kit and Yang, Kun},
  title={{{ISAC} Enhancement with Interference Exploitation: From a Uniform Viewpoint for Symbol Level and Block Level}},
  journal={IEEE Trans. Veh. Technol.}, volume={75}, number={5}, pages={7848--7862}, month={May}, year={2026}}

@article{Su2022secureDFRC,
  author={Su, Nanchi and Liu, Fan and Wei, Zhongxiang and Liu, Ya-Feng and Masouros, Christos},
  title={{Secure Dual-Functional Radar-Communication Transmission: Exploiting Interference for Resilience Against Target Eavesdropping}},
  journal={IEEE Trans. Wireless Commun.}, volume={21}, number={9}, pages={7238--7252}, year={2022}}

@article{Jia2026,
  author={Jia, Fan and Li, Ang and Liao, Xuewen and Li, Yonghui and Masouros, Christos},
  title={{Secure Precoding via Interference Exploitation in Integrated Sensing and Communication System}},
  journal={IEEE Wireless Commun. Lett.}, volume={15}, pages={425--429}, year={2026}}

@article{Su2025BCRB,
  author={Su, Nanchi and Liu, Fan and Masouros, Christos and Alexandropoulos, George C. and Xiong, Yifeng and Zhang, Qinyu},
  title={{Secure {ISAC} {MIMO} Systems: Exploiting Interference with {B}ayesian {C}ram\'er--{R}ao Bound Optimization}},
  journal={EURASIP J. Wireless Commun. Netw.}, volume={2025}, number={1}, note={Art. no. 10}, month={Feb.}, year={2025}}

@article{Bozanis2026robust,
  author={Bozanis, Dimitrios and Theologis, George and Oikonomou, Thrassos K. and Tegos, Sotiris A. and Masouros, Christos and Karagiannidis, George K.},
  title={{Secure {ISAC} Beamforming Design: A {BCRB} Approach for Target-Angle Uncertainty}},
  journal={IEEE Internet Things J.}, year={2026}, note={early access}}

@article{Suzuki2022,
  author={Fukuda, A. and Yamamoto, H. and Okazaki, H. and Suzuki, Y. and Kawai, K.},
  title={{Pinching Antenna: Using a Dielectric Waveguide as an Antenna}},
  journal={NTT DOCOMO Tech. J.}, volume={23}, number={3}, pages={5--12}, year={2022}}

@article{Ding2024flexible,
  author={Ding, Zhiguo and Schober, Robert and Poor, H. Vincent},
  title={{Flexible-Antenna Systems: A Pinching-Antenna Perspective}},
  journal={IEEE Trans. Commun.}, volume={73}, number={10}, pages={9236--9253}, month={Oct.}, year={2025}}

@article{Xu2025rate,
  author={Xu, Yanqing and Ding, Zhiguo and Karagiannidis, George K.},
  title={{Rate Maximization for Downlink Pinching-Antenna Systems}},
  journal={IEEE Wireless Commun. Lett.}, volume={14}, number={5}, pages={1431--1435}, month={May}, year={2025}}

@article{Tyrovolas2025,
  author={Tyrovolas, Dimitrios and Tegos, Sotiris A. and Diamantoulakis, Panagiotis D. and Ioannidis, Sotiris and Liaskos, Christos K. and Karagiannidis, George K.},
  title={{Performance Analysis of Pinching-Antenna Systems}},
  journal={IEEE Trans. Cogn. Commun. Netw.}, volume={12}, pages={590--601}, year={2026}}

@article{Wang2025modeling,
  author={Wang, Zhaolin and Ouyang, Chongjun and Mu, Xidong and Liu, Yuanwei and Ding, Zhiguo},
  title={{Modeling and Beamforming Optimization for Pinching-Antenna Systems}},
  journal={IEEE Trans. Commun.}, volume={73}, number={12}, pages={13904--13919}, month={Dec.}, year={2025}}

@inproceedings{Bozanis2025CRB,
  author={Bozanis, Dimitrios and Papanikolaou, Vasilis K. and Tegos, Sotiris A. and Karagiannidis, George K.},
  title={{Cram\'er-{R}ao Bounds for Integrated Sensing and Communications in Pinching-Antenna Systems}},
  booktitle={Proc. IEEE Int. Symp. Pers., Indoor Mobile Radio Commun. (PIMRC)}, address={Istanbul, T\"urkiye}, pages={1--6}, month={Sep.}, year={2025}}

@article{Li2025PASSISAC,
  author={Li, Haochen and Zhong, Ruikang and Lei, Jiayi and Liu, Yuanwei},
  title={{Pinching Antenna Systems for Integrated Sensing and Communications}},
  journal={IEEE Trans. Wireless Commun.}, volume={25}, pages={13416--13429}, year={2026}}

@article{Jiang2025BCRB,
  author={Jiang, Hao and Ouyang, Chongjun and Wang, Zhaolin and Liu, Yuanwei and Nallanathan, Arumugam and Ding, Zhiguo},
  title={{Pinching-Antenna Assisted Sensing: A {B}ayesian {C}ram\'er-{R}ao Bound Perspective}},
  journal={IEEE Trans. Commun.}, volume={74}, pages={9075--9092}, year={2026}}

@misc{Jiang2026ZZB,
  author={Jiang, Hao and Ouyang, Chongjun and Liu, Yuanwei and Nallanathan, Arumugam and Schober, Robert},
  title={{Pinching Antennas-Assisted Sensing: A {Z}iv-{Z}akai Bound ({ZZB}) Perspective}},
  howpublished={arXiv:2607.00234}, year={2026}}

@misc{Zhang2026attenuation,
  author={Zhang, Xiaochen and Du, Haitao and Cheng, Yanyu and Lin, Yushen and Teh, Kah Chan},
  title={{Physical Layer Security Performance of Pinching-Antenna Systems With In-Waveguide Attenuation}},
  howpublished={arXiv:2604.15232}, year={2026}}

@misc{Chen2026unified,
  author={Chen, Yunshu and Xue, Qing and Ouyang, Chongjun and Li, Zhidu and Wang, Yi and Hua, Meng},
  title={{Pinching-Antenna-Enabled {ISAC}: A Unified Architecture for Flexible Communication and Sensing}},
  howpublished={arXiv:2609.18083}, year={2026}}

@article{Liu2025tutorial,
  author={Liu, Yuanwei and Jiang, Hao and Xu, Xiaoxia and Wang, Zhaolin and Guo, Jia and Ouyang, Chongjun and Mu, Xidong and Ding, Zhiguo and Nallanathan, Arumugam and Karagiannidis, George K. and Schober, Robert},
  title={{Pinching-Antenna Systems ({PASS}): A Tutorial}},
  journal={IEEE Trans. Commun.}, volume={74}, pages={4881--4918}, year={2026}}

@misc{Sun2025PLS,
  author={Sun, Mengxin and Ouyang, Chongjun and Wu, Sheng and Liu, Yuanwei},
  title={{Physical Layer Security for Pinching-Antenna Systems ({PASS})}},
  howpublished={arXiv:2503.09075}, year={2025}}

@misc{Papanikolaou2025AN,
  author={Papanikolaou, Pigi P. and Bozanis, Dimitrios and Tegos, Sotiris A. and Diamantoulakis, Panagiotis D. and Sarigiannidis, Panagiotis and Karagiannidis, George K.},
  title={{Physical Layer Security with Artificial Noise in {MIMO} Pinching-Antenna Systems}},
  howpublished={arXiv:2511.23079}, year={2025}}

@article{Illi2025,
  author={Illi, Elmehdi and Qaraqe, Marwa and Ghrayeb, Ali},
  title={{Secure Pinching Antenna-aided {ISAC}}},
  journal={IEEE Commun. Lett.}, volume={30}, pages={727--731}, year={2026}}

@misc{Song2026,
  author={Song, Haowen and Zhao, Jingjing and Mu, Xidong and Cai, Kaiquan},
  title={{Pinching-Antenna Systems-enabled Secure {ISAC}: A Two-Timescale Optimization Framework}},
  howpublished={arXiv:2609.02026}, year={2026}}

@misc{Pang2026,
  author={Pang, Haoran and Wen, Miaowen and Ji, Fei and Li, Jun and Tao, Yiwei},
  title={{Antenna Placement Design for Interference Exploitation in Pinching-Antenna Systems}},
  howpublished={arXiv:2603.13929}, year={2026}}

@article{Tegos2025uplink,
  author={Tegos, Sotiris A. and Diamantoulakis, Panagiotis D. and Ding, Zhiguo and Karagiannidis, George K.},
  title={{Minimum Data Rate Maximization for Uplink Pinching-Antenna Systems}},
  journal={IEEE Wireless Commun. Lett.}, volume={14}, number={5}, pages={1516--1520}, year={2025}}

@article{Feng2026phase,
  author={Feng, Hao and Bedeer, Ebrahim and Zeng, Ming and Li, Xingwang and Gong, Shimin and Pham, Quoc-Viet},
  title={{Phase-Aware Localization in Pinching Antenna Systems: {CRLB} Analysis and {ML} Estimation}},
  journal={IEEE Commun. Lett.}, volume={30}, pages={2412--2416}, year={2026}}

@misc{Papanikolaou2026privacy,
  author={Papanikolaou, Pigi P. and Bozanis, Dimitrios and Tegos, Sotiris A. and Masouros, Christos and Karagiannidis, George K.},
  title={{Secure {ISAC} with Sensing Privacy under Eavesdropper Uncertainty}},
  howpublished={arXiv:2609.25798}, year={2026}}

@article{Boyd2011ADMM,
  author  = {S. Boyd and N. Parikh and E. Chu and B. Peleato and J. Eckstein},
  title   = {{Distributed Optimization and Statistical Learning via the Alternating Direction Method of Multipliers}},
  journal = {Found. Trends Mach. Learn.},
  volume  = {3},
  number  = {1},
  pages   = {1--122},
  year    = {2011}}

@book{Boyd2004,
  author    = {S. Boyd and L. Vandenberghe},
  title     = {{Convex Optimization}},
  publisher = {Cambridge University Press},
  address   = {Cambridge, U.K.},
  year      = {2004}
}

@article{Cui2022,
  author={M. Cui and L. Dai},
  title={{Channel Estimation for Extremely Large-Scale {MIMO}: Far-Field or Near-Field?}},
  journal={IEEE Trans. Commun.},
  volume={70},
  number={4},
  pages={2663--2677},
  month={Apr.},
  year={2022}}

@article{NF,
  author={Wang, Zhaolin and Mu, Xidong and Liu, Yuanwei},
  title={{Near-Field Integrated Sensing and Communications}},
  journal={IEEE Commun. Lett.}, volume={27}, number={8}, pages={2048--2052}, year={2023}}
